%% file: lanczos_spec.tex
\documentclass{article}

\usepackage{microtype}
\usepackage{graphicx}
\usepackage{subfigure}
\usepackage{booktabs} % for professional tables

\usepackage{hyperref}

\usepackage[accepted]{icml2021}

\icmltitlerunning{Analysis of stochastic Lanczos quadrature for spectrum approximation}

\newcommand{\MYRETURN}{\item[\textbf{return}]}

\usepackage{afterpage}
\usepackage{setspace}

\usepackage{enumitem}
\usepackage{amsmath,amssymb,amsthm}

\usepackage{thmtools}
\usepackage{thm-restate}
\usepackage{xparse}

\newtheorem{definition}{Definition}
\newtheorem{lemma}{Lemma}

\newtheorem{proposition}{Proposition}
\newtheorem{theorem}{Theorem}

\theoremstyle{remark}
\newtheorem{remark}{Remark}

\newtheoremstyle{citedremark}%
  {3pt}% (space above)
  {3pt}% (space below)
  {}% (body font)
  {}% (indent amount)
  {\itshape}% {theorem head font}
  {.}% {punctuation after theorem head}
  {.5em}% {space after theorem head}
  {\thmname{#1} \thmnumber{#2} {\thmnote{\normalfont#3}}}% {theorem head spec}

\theoremstyle{citedremark}
\newtheorem{cremark}[remark]{Remark}

\newtheoremstyle{cited}%
  {3pt}% (space above)
  {3pt}% (space below)
  {\itshape}% (body font)
  {}% (indent amount)
  {\bfseries}% {theorem head font}
  {.}% {punctuation after theorem head}
  {.5em}% {space after theorem head}
  {\thmname{#1} \thmnumber{#2} {\thmnote{\normalfont#3}}}% {theorem head spec}

\theoremstyle{cited}
\newtheorem{ctheorem}[theorem]{Theorem}

\usepackage{bbm}
\newcommand{\bOne}{\ensuremath{\mathbbm{1}}}

\newcommand{\EE}{\ensuremath{\mathbb{E}}}

\newcommand{\PP}{\ensuremath{\mathbb{P}}}

\renewcommand{\d}[1]{\ensuremath{\mathrm{d}#1}}
\renewcommand{\vec}{\mathbf}
\newcommand{\R}{\mathbb{R}}
\newcommand{\tr}{\operatorname{tr}}
\newcommand{\T}{\textup{\textsf{T}}}
\newcommand{\ehat}{\hat{\vec{e}}}

\newcommand{\samp}[1]{\langle #1 \rangle}

\newcommand{\gq}[2]{[#2]_{#1}^{\textup{gq}}}
\newcommand{\gqupper}[2]{[#2]_{#1}^{\uparrow}}
\newcommand{\gqlower}[2]{[#2]_{#1}^{\downarrow}}

\newcommand{\W}{d_{\mathrm{W}}}
\newcommand{\KS}{d_{\mathrm{KS}}}

\newcommand{\mnote}[1]{}

\newcommand{\xnote}[1]{}

\usepackage[nameinlink,capitalise,noabbrev]{cleveref}
\crefname{ctheorem}{Theorem}{Theorems}
\crefname{appendix}{Supplement}{Supplements}

\begin{document}

\twocolumn[
\icmltitle{Analysis of stochastic Lanczos quadrature for spectrum approximation}

% It is OKAY to include author information, even for blind
% submissions: the style file will automatically remove it for you
% unless you've provided the [accepted] option to the icml2021
% package.

% List of affiliations: The first argument should be a (short)
% identifier you will use later to specify author affiliations
% Academic affiliations should list Department, University, City, Region, Country
% Industry affiliations should list Company, City, Region, Country

% You can specify symbols, otherwise they are numbered in order.
% Ideally, you should not use this facility. Affiliations will be numbered
% in order of appearance and this is the preferred way.
\icmlsetsymbol{equal}{*}

\begin{icmlauthorlist}
\icmlauthor{Tyler Chen}{uw}
\icmlauthor{Thomas Trogdon}{uw}
\icmlauthor{Shashanka Ubaru}{ibm}
\end{icmlauthorlist}

\icmlaffiliation{uw}{Department of Applied Mathematics, University of Washington, Seattle, Washington, USA}
\icmlaffiliation{ibm}{IBM T.J. Watson Research Center, Yorktown Heights, New York, USA}

\icmlcorrespondingauthor{Tyler Chen}{chentyl@uw.edu}

% You may provide any keywords that you
% find helpful for describing your paper; these are used to populate
% the "keywords" metadata in the PDF but will not be shown in the document
\icmlkeywords{spectrum approximation, Lanczos, trace estimation}

\vskip 0.3in
]

% this must go after the closing bracket ] following \twocolumn[ ...

% This command actually creates the footnote in the first column
% listing the affiliations and the copyright notice.
% The command takes one argument, which is text to display at the start of the footnote.
% The \icmlEqualContribution command is standard text for equal contribution.
% Remove it (just {}) if you do not need this facility.

\printAffiliationsAndNotice{}  % leave blank if no need to mention equal contribution
%\printAffiliationsAndNotice{\icmlEqualContribution} % otherwise use the standard text.

\setlength{\abovedisplayskip}{.75em}%
\setlength{\belowdisplayskip}{.75em}%
\setlength{\abovedisplayshortskip}{.5em}%
\setlength{\belowdisplayshortskip}{.5em}%

\input{lanczos_spec_body.tex}

\clearpage

\bibliographystyle{icml2021}
\bibliography{lanczos_spec}

%%%%%%%%%%%%%%%%%%%%%%%%%%%%%%%%%%%%%%%%%%%%%%%%%%%%%%%%%%%%%%%%%%%%%%%%%%%%%%%
%%%%%%%%%%%%%%%%%%%%%%%%%%%%%%%%%%%%%%%%%%%%%%%%%%%%%%%%%%%%%%%%%%%%%%%%%%%%%%%
% DELETE THIS PART. DO NOT PLACE CONTENT AFTER THE REFERENCES!
%%%%%%%%%%%%%%%%%%%%%%%%%%%%%%%%%%%%%%%%%%%%%%%%%%%%%%%%%%%%%%%%%%%%%%%%%%%%%%%
%%%%%%%%%%%%%%%%%%%%%%%%%%%%%%%%%%%%%%%%%%%%%%%%%%%%%%%%%%%%%%%%%%%%%%%%%%%%%%%

\clearpage

\onecolumn
\appendix

\input{lanczos_spec_appendix.tex}

%%%%%%%%%%%%%%%%%%%%%%%%%%%%%%%%%%%%%%%%%%%%%%%%%%%%%%%%%%%%%%%%%%%%%%%%%%%%%%%
%%%%%%%%%%%%%%%%%%%%%%%%%%%%%%%%%%%%%%%%%%%%%%%%%%%%%%%%%%%%%%%%%%%%%%%%%%%%%%%

\end{document}

%% file: lanczos_spec_body.tex
\begin{abstract}
    The cumulative empirical spectral measure (CESM) \( \Phi[\vec{A}] : \mathbb{R} \to [0,1] \) of a \( n\times n \) symmetric matrix \( \vec{A} \) is defined as the fraction of eigenvalues of \( \vec{A} \) less than a given threshold, i.e., \( \Phi[\vec{A}](x) := \sum_{i=1}^{n} \frac{1}{n} \bOne[ \lambda_i[\vec{A}]\leq x] \).
Spectral sums \( \tr(f[\vec{A}]) \) can be computed as the Riemann--Stieltjes integral of \( f \) against \( \Phi[\vec{A}] \), so the task of estimating CESM arises frequently in a number of applications, including machine learning. 
We present an error analysis for stochastic Lanczos quadrature (SLQ).
    We show that SLQ obtains an approximation to the CESM within a Wasserstein distance of \( t \: | \lambda_{\textup{max}}[\vec{A}] - \lambda_{\textup{min}}[\vec{A}] | \) with probability at least \( 1-\eta \), by applying the Lanczos algorithm for \( \lceil 12 t^{-1} + \frac{1}{2} \rceil  \) iterations to \( \lceil 4 ( n+2 )^{-1}t^{-2} \ln(2n\eta^{-1}) \rceil \) vectors sampled independently and uniformly from the unit sphere.
We additionally provide (matrix-dependent) a posteriori error bounds for the Wasserstein and Kolmogorov--Smirnov distances between the output of this algorithm and the true CESM.
The quality of our bounds is demonstrated using numerical experiments. 
\end{abstract}

\section{Introduction}

Given an \( n\times n \) symmetric matrix \( \vec{A} \), the cumulative empirical spectral measure (CESM) \( \Phi[\vec{A}] : \R \to [0,1] \) gives the fraction of eigenvalues less than a given threshold. 
That is,
\begin{align*}
    \Phi[\vec{A}](x) 
    := %\int_{-\infty}^{x} \phi(s;\vec{A}) \d{s} 
    \sum_{i =1}^{n} \frac{1}{n}\bOne[ \lambda_i[\vec{A}] \leq x ],
\end{align*}
where \( \bOne[ \: \cdot \leq x] : \mathbb{R} \to \{0,1\} \) is the indicator function defined by \( \bOne[s \leq x] = 1 \) if \( s\leq x \) and \( \bOne[s \leq x] = 0 \) if \( s > x \).
The CESM contains all information about spectrum of \( \vec{A} \) and can therefore be used to compute any quantity depending on just the spectrum.
Conversely, computing the CESM requires exact knowledge of all the eigenvalues of \( \vec{A} \), which are expensive to compute.

For many applications, however, it suffices to provide a coarse estimate of the CESM.
In machine learning, approximate CESMs have found use in facilitating backpropagation through implicit likelihoods \cite{ramesh_lecun_18} as well as for for studying properties of Hessians during neural network training \cite{ghorbani_krishnan_xiao_19,papyan_19,yao_gholami_keutzer_mahoney_20}.
This provides insight into differences between various training approaches and/or network architectures.

In data science more broadly, approximate CESMs have become a popular approach for exploring properties of graphs and networks as well as for approximating fundamental quantities such as matrix norms, log-determinants, number of eigenvalues in an interval, etc. \cite{avron_10,dinapoli_polizzi_saad_16,ubaru_saad_seghouane_17,han_malioutov_avron_shin_17,xi_li_saad_18,musco_netrapalli_sidford_ubaru_woodruff_19,dong_benson_bindel_19}.
Approximate CESMs have also long been used in computational physics and chemistry to study various observables \cite{ducastelle_cryotlackmann_70,haydock_heine_kelly_75,wheeler_blumstein_72} and remain widely used in these fields today \cite{kpm_review_06,covaci_peeters_berciu_10,sbierski_trescher_bergholtz_brouwer_17,schnack_richter_steinigeweg_20}.

Finally, and as a result of their general usefulness, approximate CESMs have become the first stage of a range of algorithms for fundamental linear algebraic tasks including methods for computing matrix functions \cite{fan_shuman_ubaru_saad_19} and state of the art parallel eigensolvers \cite{polizzi_09,evsl_19}.
%Finally, approximate CESMs have long been used in computational physics and chemistry \cite{ducastelle_cryotlackmann_70,haydock_heine_kelly_75,wheeler_blumstein_72} where they remain of widespread use \cite{kpm_review_06,covaci_peeters_berciu_10,sbierski_trescher_bergholtz_brouwer_17,schnack_richter_steinigeweg_20}.
\iffalse
Such estimates found early use in computational physics and chemistry \cite{ducastelle_cryotlackmann_70,haydock_heine_kelly_75,wheeler_blumstein_72} where they remain highly relevant \cite{kpm_review_06,covaci_peeters_berciu_10,sbierski_trescher_bergholtz_brouwer_17,schnack_richter_steinigeweg_20}.
Outside of these fields, approximate CESMs have become a popular approach for approximating various matrix norms, log-determinants, Estrada indices, triangle counts in a graph \cite{avron_10,ubaru_saad_seghouane_17,han_malioutov_avron_shin_17,musco_netrapalli_sidford_ubaru_woodruff_19}, as well as related tasks such as studying network motifs \cite{dong_benson_bindel_19}, estimating the number of eigenvalues in an interval \cite{dinapoli_polizzi_saad_16,xi_li_saad_18}, and recently, for studying properties of Hessians during neural network training \cite{ghorbani_krishnan_xiao_19,papyan_19,yao_gholami_keutzer_mahoney_20}.

\mnote{
\fi

In this paper, we consider a well-known algorithm for computing an approximate CESM, which we refer to as stochastic Lanczos quadrature (SLQ).
The algorithm described in this paper and closely related methods have been used for decades to estimate spectral densities \cite{haydock_heine_kelly_72,haydock_heine_kelly_75,lambin_gaspard_82,benoit_royer_poussigue_92}, and like the Lanczos algorithm \cite{lanczos_50} on which they are based, they remain highly relevant today \cite{lin_saad_yang_16,ghorbani_krishnan_xiao_19,papyan_19}.

%The analysis of such algorithms is thus broken into two corresponding parts: (i) understanding the sample complexity required so that averages of the unbiased estimator for \( \Phi[\vec{A}](x) \) provide a sufficiently accurate sample average, and (ii) understanding how many iterations of the iterative method are required to sufficiently approximate each sample.

\subsection{Stochastic Lanczos Quadrature}

Using the standard definition of a matrix function,\hspace{-.15em}%
\footnote{For a symmetric matrix \( \vec{A} \) with eigenvalue decomposition \( \vec{A} = \vec{U} \vec{\Lambda} \vec{U}^\T \) and scalar function \( f: \R \to \R \), the matrix function \( f[\vec{A}] \) is defined by \( f[\vec{A}] := \vec{U} f[\vec{\Lambda}] \vec{U}^\T \), where \( f[\vec{\Lambda}] \) is the diagonal matrix with \( f \) applied entrywise to the diagonal entries of the matrix \( \vec{\Lambda} \).} %
for a symmetric matrix \( \vec{A} \), we denote by \( \bOne[\vec{A} \leq x] \) the matrix with the same eigenvectors as \( \vec{A} \), but whose eigenvalues are 0 or 1 depending on whether the corresponding eigenvalue of \( \vec{A} \) is below or above \( x \);
that is, \( \bOne[\vec{A} \leq x] \) is the orthogonal projector onto the eigenspace associated to eigenvalues \( \lambda_i[\vec{A}] \) such that \( \lambda_i[\vec{A}] \leq x \).
Thus, \( \Phi[\vec{A}](x) = n^{-1} \tr(\bOne[\vec{A} \leq x]) \).

With this definition in place, we define the weighted CESM,
\begin{align*}
    %\label{eqn:intro_estimator}
    \Psi[\vec{A},\vec{v}](x) := \vec{v}^\T \bOne[\vec{A} \leq x] \vec{v}.
\end{align*}
If \( \vec{v} \sim \mathcal{U}(\mathcal{S}^{n-1}) \), where \( \mathcal{U}(\mathcal{S}^{n-1}) \) is the uniform distribution on the unit sphere, then the weighted CESM has the desirable properties (i) that it is an unbiased estimator for \( \Phi[\vec{A}](x) \), and (ii) that it defines a cumulative probability distribution function; i.e. \( \EE[\Psi[\vec{A},\vec{v}](x)] = \Phi[\vec{A}](x) \) and \( \Psi[\vec{A},\vec{v}] : \R \to [0,1] \) is weakly increasing, right continuous, and
\begin{align*}
    \lim_{x\to-\infty} \Psi[\vec{A},\vec{v}](x) = 0%\vec{v}^\T\bOne[\vec{A} \leq x]\vec{v} = 0
    ,&&
    \lim_{x\to\infty} \Psi[\vec{A},\vec{v}](x) = 1.%\vec{v}^\T\bOne[\vec{A} \leq x]\vec{v} = 1,
\end{align*}
%These properties, along with concentration as the matrix size \( n \) increases, are illustrated in \cref{fig:wCESM}.

Next, we consider the degree \( k \) Gaussian quadrature rule \( \gq{k}{\Psi[\vec{A},\vec{v}]} \) for \( \Psi[\vec{A},\vec{v}] \).
In general, a Gaussian quadrature rule for a distribution function can be computed using the Stieltjes procedure, which for distributions of the form \( \Psi[\vec{A},\vec{v}] \), is equivalent to the Lanczos algorithm \cite{gautschi_04,golub_meurant_09}.
Specifically, if \( [\vec{T}]_{:k,:k} \) is the symmetric tridiagonal matrix obtained by running Lanczos on \( \vec{A} \) and \( \vec{v} \) for \( k \) steps, then
\begin{align*}
    \gq{k}{\Psi[\vec{A},\vec{v}]} = \Psi([\vec{T}]_{:k,:k},\ehat)
\end{align*}
where \( \ehat = [1,0,\ldots, 0]^\T \).

By repeating this process over multiple samples and averaging, we arrive at SLQ, outlined in \cref{alg:SLQ}.
\begin{algorithm}[htb!]
    \caption{Stochastic Lanczos Quadrature}\label{alg:SLQ}
    \begin{algorithmic}
    \INPUT {$\vec{A}$, $n_{\textup{v}}$, $k$}
    \FOR {\( i=1,2,\ldots, n_{\textup{v}} \)}
        \STATE Sample \( \vec{v}_i \sim \mathcal{U}(\mathcal{S}^{n-1}) \) (and define \( \Psi_i = \Psi(\vec{A},\vec{v}_i) \))
        \STATE Run Lanczos on \( \vec{A},\vec{v}_i \) for \( k \) steps to compute \( [\vec{T}_i]_{:k,:k} \)

        \STATE Define  \( \gq{k}{\Psi_i} = \Psi[[\vec{T}_i]_{:k,:k},\ehat] \)%\( \displaystyle \gq{k}{\Psi_i}(x) = \sum_{j=1}^{k} [d_i]_j \bOne \left[ [\theta_i]_j \leq x \right] \)
    \ENDFOR
    \MYRETURN \( \samp{ \gq{k}{\Psi_i} } := \frac{1}{n_{\textup{v}}} \sum_{i=1}^{n_{\textup{v}}} \gq{k}{\Psi_i} \)
\end{algorithmic}
\end{algorithm}

SLQ is computationally efficient.
In particular, all samples can be computed in parallel on separate machines or on a single machine using blocked linear algebra routines.
Moreover, the algorithm is matrix free in that we only require a method to compute the map \( \vec{v} \mapsto \vec{A} \vec{v} \), the cost of which we denote by \( T_{\textup{mv}} \).
This is particularly important for large dense matrices, where the \( O(n^2) \) storage required to keep every entry of \( \vec{A} \) may be intractable.
In many cases, such as with Hessians from the training of neural networks, matrix-vector products can be computed implicitly using only \( O(n) \) storage \cite{pearlmutter_94}.
Similarly, if \( \vec{A} \) is sparse or structured, this map may be evaluated faster than the \( O(n^2) \) computation cost for an arbitrary matrix vector product; e.g. \( O( \operatorname{nnz}[\vec{A}] ) \) for a sparse matrix or \( O(n\ln n) \) for a circulant matrix.

\subsection{Discussion of results}

Our first main result is a runtime guarantee for SLQ.
In particular, we show that if \( n_{\textup{v}} > 4 ( n+2 )^{-1} t^{-2} \ln(2n\eta^{-1}) \) and \( k > 12 t^{-1} + \frac{1}{2} \), then 
\begin{align*}
    \PP\big[ \W( \Phi[\vec{A}], \samp{\gq{k}{\Psi_i} } ) > t I[\vec{A}]  \big] < \eta,
\end{align*}
where \( I[\vec{A}] := | \lambda_{\textup{max}}[\vec{A}] - \lambda_{\textup{min}}[\vec{A}] | \) and \( \W(\cdot,\cdot) \) denotes the Wasserstein distance between two distribution functions as in \cref{def:wasserstein} below.
This implies that as \( n\to\infty \), for \( t \gg  n^{-1/2}  \), SLQ has a runtime of \( O( T_{\textup{mv}} t^{-1} \log( t^{-2} \eta^{-1}) ) \).
This bound is nearly tight in the sense that for any \( t \in (0,1) \), there exists a matrix of size \( \lceil (4t)^{-1} \rceil \) such that at least \( (8t)^{-1} \) matrix vector products are required to obtain an output with Wasserstein distance less than \( t I[\vec{A}] \).

The second main result is an a posteriori upper bound for Wasserstein and Kolmogorov--Smirnov (KS) distances, which take into account spectrum dependent features such as clustered or isolated eigenvalues.

Finally, in proving these results, we show that if \( n_{\textup{v}} > (n+2)^{-1} t^{-2}  \ln(2 \eta^{-1}) \) then, for any \( x\in\R \),
\begin{align}
    \label{eqn:intro_sample_ave}
    \PP \left[ \left| \Phi[\vec{A}](x) - \left( \frac{1}{n_{\textup{v}}}\sum_{i=1}^{n_{\textup{v}}} \vec{v}_i^\T \bOne[\vec{A} \leq x] \vec{v}_i \right) \right| > t \right] < \eta.
\end{align}
This is applicable to the analysis of a range of algorithms beyond SLQ.

\subsection{Related work}

\mnote{this is really three fairly separate points: other spectrum approximation methods, the quality of \( \Psi[\vec{A},\vec{v}] \), and relation to other SLQ. Not sure about the order}

A variety of algorithms for approximating the CESM have been developed; see \cite{fischer_11,kpm_review_06,lin_saad_yang_16,large_matrices_density_review_18,cohensteinen_kong_sohler_valiant_18} and the references therein for a more complete overview.
By far, the most popular algorithms are the kernel polynomial method (KPM) \cite{kpm_review_06} and SLQ.
The two algorithms differ primarily in how they approximate the weighted CESM \( \Psi[\vec{A},\vec{v}] \), and as a result, our analysis of the estimator \( \Psi[\vec{A},\vec{v}] \) for \( \Phi[\vec{A}]\) applies to the KPM as well.
The main advantages of SLQ are that it is adaptive to the spectrum of \( \vec{A} \) (i.e. it automatically detects features such as gaps in the spectrum, isolated or clustered eigenvalues, etc.) and that it produces an output which is a probability distribution function.
However, in doing so, SLQ requires the computation of inner products, which may be costly, especially in a distributed computing environment \cite{lin_saad_yang_16}.
Moreover, the behavior of the Lanczos method is not always straightforward in finite precision.

We remark that the KPM, both with exact and inexact matrix-vector products, has recently been analyzed \cite{braverman_krishnan_musco_21}.
The analysis for KPM with exact matrix-vector products yields similar rates to our analysis for SLQ, but the sample complexity is given in terms of unspecified universal constants and has a polynomially worse dependence on the accuracy parameter \( t \).
We believe this is an artifact of analysis not an inherent property of the KPM, but to be certain a full side-by-side comparison of the algorithms is needed.

Tail bounds similar to \eqref{eqn:intro_sample_ave} have been derived in several contexts.
First, while explicit constants are not specified, the same \( n^{-1} t^{-2} \ln(\eta^{-1}) \) dependence for \( n_{\textup{v}} \) is implied by \citep[Lemma 4.5]{deift_trogdon_20}.
Second, if the \( \vec{v}_i \) are replaced by \emph{unnormalized} Gaussian vectors (with mean zero and variance \( n^{-1} \)), then well known bounds for trace estimation \cite{avron_toledo_11,roostakhorasani_ascher_14} yield similar rates in terms of explicit constants. 
However, the weighed CESM corresponding to an unnormalized sample will not be a probability distribution function.

Finally, the algorithm studied in this paper is closely related to an algorithm, also commonly referred to as SLQ (but called gSLQ in this paper to avoid confusion), for approximating spectral sums \( \tr( f[\vec{A}] ) = \sum_{i=1}^{n} f(\lambda_i[\vec{A}]) \) \cite{bai_fahey_golub_96,bai_golub_96}.
Indeed, the SLQ studied here is a special case of gSLQ, with \( f(s) = \bOne[s\leq x] \).
However, SLQ is in fact equivalent to gSLQ the sense that the gSLQ approximation to \( \tr(f[\vec{A}]) \) can be obtained by computing a Riemann--Stieltjes  integral of \( f \) against the output of \cref{alg:SLQ}.
More generally, matrix function trace estimation is closely related to CESM estimation due to the fact that
\begin{align*}
    \tr( f[\vec{A}] )
    = n \int f(s) \d\Phi(s).
\end{align*}

In \cite{ubaru_chen_saad_17}, the convergence of gSLQ when \( f \) is smooth or analytic is studied.
As a result, this analysis is not immediately applicable to CESM estimation itself, since \( \bOne[\:\cdot \leq x] \) is discontinuous.
One possibility is to solve a relaxed problem where \( \bOne[\:\cdot \leq x] \) is replaced with a smoothed approximation such as a shifted hyperbolic tangent or the CDF of any continuous random variable with small enough variance.
This results in an approximation to the CESM equivalent to the convolution of the CESM with a smoothing kernel \cite{ubaru_chen_saad_17,han_malioutov_avron_shin_17,ghorbani_krishnan_xiao_19}.

If the CESM can be reasonably smoothed, then such an approach works well. 
However, it is often the case that the ``variance'' of the smoothing kernel, in order to preserve certain aspects of the CESM, such as large jumps due to clustered eigenvalues, has to be very small.
In such cases, Gaussian quadrature bounds for smooth functions are often useless, necessitating bounds such as the ones presented in this paper; see \cref{sec:GQ_smooth_bounds} for a detailed discussion.

\subsection{Preliminaries}

Matrices are denoted by bold uppercase letters and vectors are denoted by bold lowercase letters.
The first canonical unit vector \( [1,0,\ldots, 0]^\T \), of size determined by context, is denoted \( \ehat \).
The set of all eigenvalues of a \( d\times d \) symmetric matrix \( \vec{B} \) is denoted \( \Lambda[\vec{B}] \), and the individual eigenvalues are \( \lambda_{\textup{min}}[\vec{B}] = \lambda_d[\vec{B}] \leq \cdots \leq \lambda_1[\vec{B}] = \lambda_{\textup{max}}[\vec{B}] \).
Unless otherwise stated, \( \vec{A} \) is \( n\times n \) symmetric matrix.

We denote the \( i \)-th entry of a vector \( \vec{v} \) by \( [\vec{v}]_i \) and the submatrix consisting of rows \( r \) to \( r' \) and columns \( c \) to \( c' \) by \( [\vec{B}]_{r:r',c:c'} \). 
If any of these indices are equal to \( 1 \) or \( n \), they may be omitted.
If \( r = r' \) or \( c = c' \), then we will simply write this index once.
Thus, \( [\vec{B}]_{:,:2} \) denotes the first two columns of \( \vec{B} \), and \( [\vec{B}]_{3,:} \) denotes the third row of \( \vec{B} \).

For some positive integer \( n_{\textup{v}} \) and a set of values \( \{ x_i \}_{i=1}^{n_{\textup{v}}} \), the sample average \( \frac{1}{n_{\textup{v}}}\sum_{i=1}^{n_{\textup{v}}} x_i \) is denoted \( \samp{x_i} \).

\begin{definition}
Let \( \mu \) and \( \nu \) be two probability distribution functions.
We say the moments of \( \mu \) and \( \nu \) are equal up to degree \( k-1 \) if for all polynomials \( p \) of degree \( < k \),
\begin{align*}
    \int p(x) \d{\mu(x)}
    = \int p(x) \d{\nu(x)}.
\end{align*}
\end{definition}

We also have the standard definition of Kolmogorov-Smirnov and Wasserstein distances.
\begin{definition}
Let \( \mu \) and \( \nu \) be two probability distribution functions.
The Kolmogorov--Smirnov distance between \( \mu \) and \( \nu \), denoted \( \KS(\mu,\nu) \), is defined by
\begin{align*}
    \KS(\mu,\nu) := \sup_{x} | \mu(x) - \nu(x) |.
\end{align*}
\end{definition}
\begin{definition}
\label{def:wasserstein}
Let \( \mu \) and \( \nu \) be two probability distribution functions.
The Wasserstein (earth mover) distance between \( \mu \) and \( \nu \), denoted \( \W(\mu,\nu) \), is defined by
\begin{align*}
    \W(\mu,\nu) := \int | \mu(x) - \nu(x) | \d{x}.
\end{align*}
\end{definition}

\section{The Lanczos algorithm}

The primary computational cost of SLQ is due to the Lanczos algorithm \cite{lanczos_50}.
The Lanczos algorithm is typically viewed as a procedure for constructing an orthonormal basis \( [\vec{Q}]_{:,:k} := [ \vec{q}_1 , \ldots , \vec{q}_k ] \) for the Krylov subspace
\begin{align*}
\mathcal{K}_k(\vec{A},\vec{v})
= \operatorname{span}( \vec{v}, \vec{A}\vec{v}, \ldots, \vec{A}^{k-1}\vec{v} ).
%= \{ p[\vec{A}] \vec{v} : \deg(p) < k \}.
\end{align*}
This can be done by a Gram--Schmidt-like process, where \( \vec{A} \vec{q}_{k} \) is orthogonalized against previous basis vectors \( \{ \vec{q}_j \}_{j=1}^{k} \), which results in a factorization
\begin{align*}
%    \label{eqn:lanczos_factorization}
    \vec{A} [\vec{Q}]_{:,:k} = [\vec{Q}]_{:,:k} [\vec{T}]_{:k,:k} + \beta_k \vec{q}_{k+1} \vec{e}_k^{\T}
\end{align*}
where \( [\vec{T}]_{:k,:k} \) is upper-Hessenberg.
However, since \( [\vec{T}]_{:k,:k} = [\vec{Q}]_{:,:k}^\T \vec{A} [\vec{Q}]_{:,:k} \) is symmetric, then \( [\vec{T}]_{:k,:k} \) is actually tridiagonal.
Thus, \( \vec{A} \vec{q}_k \) will be orthogonal to \( \vec{q}_j \), \( j < k -1 \), so we only need to orthogonalize against \( \vec{q}_k \) and \( \vec{q}_{k-1} \) in each iteration.
As a result, the runtime of \cref{alg:lanczos} is \( O(k (T_{\textup{mv}} + n ) ) \) and the required storage is \( O(n) \).

\begin{algorithm}[h]
\caption{Lanczos}\label{alg:lanczos}
    \begin{algorithmic}
    \INPUT {$\vec{A}$,$\vec{v}$,$k$}
    \STATE \( \vec{q}_0 = \vec{0} \),
    \( \beta_{-1} = 0 \),
    \( \vec{q}_1  = \vec{v} / \| \vec{v} \| \)
    \FOR {\( i=0,1,\ldots, k-1 \)}
        \STATE \( \tilde{\vec{q}}_{i+1} = \vec{A} \vec{q}_{i} - \beta_{i-1} \vec{q}_{i-1} \)
        \STATE \( \alpha_i = \langle \tilde{\vec{q}}_{i+1}, \vec{q}_i \rangle \)
        \STATE \( \tilde{\vec{q}}_{i+1} = \tilde{\vec{q}}_{i+1} - \alpha_i \vec{q}_i \)
        %\STATE if desired, reorthogonalize \( \tilde{\vec{q}}_{i+1} \) against \( \vec{q}_1, \ldots, \vec{q}_{i-1} \)
        \STATE \algorithmicif~`reorthogonlization'~\algorithmicthen~
        orthogonalize \( \tilde{\vec{q}}_{i+1} \) against \( \{\vec{q}_j\}_{j=1}^{i-1} \)
        %\( \tilde{\vec{q}}_{i+1} = (\vec{I} - [\vec{Q}]_{:,:i-1} [\vec{Q}]_{:,:i-1}^\T) \tilde{\vec{q}}_{i+1} \)
        \algorithmicendif \label{alg:reorth_line}
        %\FOR {\( i = 1,2,\ldots,i \)}
        %    \STATE \( \tilde{\vec{q}}_{i+1} = \tilde{\vec{q}}_{i+1} - \langle \tilde{\vec{q}}_{i+1} , \vec{q}_i \rangle \vec{q}_i \)
        %\ENDFOR
        \STATE \( \beta_{i} = \| \vec{q}_{i+1} \| \)
        \STATE \( \vec{q}_{i+1} = \tilde{\vec{q}}_{i+1} / \beta_{i} \)
    \ENDFOR
    \MYRETURN \( [\vec{T}]_{:k,:k} \) (diagonal \( [\alpha_1, \ldots, \alpha_k] \) and the sub and super diagonal \( [\beta_1, \ldots, \beta_{k-1}] \))% (\( [\vec{T}]_{i,i} := \alpha_i \), \( [\vec{T}]_{i,i+1}, [\vec{T}]_{i+1,i} := \beta_i \))% and optionally \( [\vec{Q}]_{:,:k} = [\vec{q}_1, \ldots, \vec{q}_k] \)
\end{algorithmic}
\end{algorithm}

\begin{remark}
\label{thm:exact_lanczos}
    If \cref{alg:lanczos} is run for \( k = n \) iterations on any right hand side with nonzero projection onto each eigenvector, the tridiagonal matrix \( [\vec{T}]_{:n,:n} \) produced will have the same eigenvalues as \( \vec{A} \).
Thus the CESM can be computed deterministically in time \( O( n T_{\textup{mv}} + n^2 ) \).
\end{remark}

\begin{cremark}[{\citep{gu_eisenstat_95}}]
The eigenvalues and first components of eigenvectors of a real symmetric tridiagonal matrix of size \( k\times k \) can be computed in \( O(k^2) \) operations.
\end{cremark}

\begin{remark}
    Without reorthogonalization in \cref{alg:lanczos}, the runtime of SLQ \cref{alg:SLQ} is \( O( n_{\textup{v}}k (T_{\textup{mv}} + n) ) \) and the required storage is \( O(n) \).
    With reorthogonalization, the runtime is \( O( n_{\textup{v}} k ( T_{\textup{mv}}  + nk) ) \) and the required storage is \( O(nk) \).
\end{remark}

\begin{remark}
    In exact arithmetic, the reorthogonalization step of \cref{alg:lanczos} is unnecessary as \( \tilde{\vec{q}}_{i+1} \) is already orthogonal to \( \{ \vec{q}_j \}_{j=1}^{i-1} \).
However, in finite precision arithmetic this orthogonality may be lost.
Our bounds, as well as the bounds for gSLQ \cite{ubaru_chen_saad_17}, are derived based on exact arithmetic theory, so it must be assumed that \( [\vec{T}]_{:k,:k} \) is computed using some implementation which produces an output close to the exact arithmetic output.
The easiest way to ensure this is with full reorthogonalization, although other more advanced schemes have been considered.

For the task of computing CESMs, some practitioners \cite{papyan_19} have noted the algorithm still works without reorthogonalization.
In \cref{sec:finite_precision}, we provide an overview of existing analysis on the Lanczos algorithm in finite precision \cite{paige_76,paige_80,greenbaum_89} which provides a high level explanation as to why SLQ still works without reorthogonalization.
\end{remark}

\section{Results}

We now state the main results. 
\begin{restatable}{theorem}{wassersteinapriori}
\label{thm:wasserstein_apriori}
Given \( 0<\eta<1 \) and \( t>0 \), set 
    \( n_{\textup{v}} > 4 ( n+2 )^{-1} t^{-2} \ln(2n\eta^{-1}) \) and \( k > 12 t^{-1} + \frac{1}{2} \).
Then, \cref{alg:SLQ} will output an estimate \( \samp{ \gq{k}{ \Psi_i } } \) satisfying, %with failure probability at most \( \eta \),
\begin{align*}
    \PP\big[ \W( \Phi[\vec{A}], \samp{\gq{k}{\Psi_i} } ) > t I[\vec{A}]  \big] < \eta.
\end{align*}
where \( I[\vec{A}] := | \lambda_{\textup{max}}[\vec{A}] - \lambda_{\textup{min}}[\vec{A}] | \).
\end{restatable}
\iffalse
Note that the number of samples required by \cref{thm:wasserstein_apriori} decays with the matrix size \( n \).
In particular, if \( t = \Omega( n^{-1/2} ) \), we require only a small number of samples; \( n_{\textup{v}} = O( \ln(n \eta^{-1}) ) \).
In this regime, we have a runtime roughly \( O(t^{-1} (T_{\textup{mv}} + n) \ln(n \eta^{-1})) \) without reorthognalization, or \( O(t^{-1}(T_{\textup{mv}} + t^{-1} n) \ln(n \eta^{-1})) \) with reorthogonalization.
Assuming \( T_{\textup{mv}} = \Omega(n) \), the runtime is \( \tilde{O}(t^{-1} T_{\textup{mv}} ) \) in both cases. 
\fi
\Cref{thm:wasserstein_apriori} is essentially a direct consequence of following theorem of the average of weighted CESMs and a straightforward bound on the Wasserstein distances of distribution functions with matching moments.

\begin{restatable}{theorem}{estimatorprob}
\label{thm:CESM_estimator_prob}
Given a positive integer \( n_{\text{v}} \), suppose \( \{ \vec{v}_i \}_{i=1}^{n_{\textup{v}}} \stackrel{\textup{iid}}{\sim} \mathcal{U}(\mathcal{S}^{n-1}) \) and define \( \Psi_i = \Psi(\vec{A},\vec{v}_i) \).
Then, for all \( x \in \R \) and \( t > 0 \),
\begin{align*}
    \PP\big[ \left| \Phi[\vec{A}](x) - \samp{ \Psi_i(x) } \right| > t \big]
    &\leq
    2 \exp\left( - n_{\textup{v}} (n+2) t^2 \right)
    \\
    \PP \big[ \KS (\Phi[\vec{A}], \samp{ \Psi_i } ) > t \big]
    &\leq
    2 n \exp\left( - n_{\textup{v}} (n+2) t^2 \right).
\end{align*}
\end{restatable}

\begin{restatable}{proposition}{wassersteinmoments}
\label{thm:wasserstein_moments}
Suppose \( \mu \) and \( \nu \) are two probability distribution functions constant on the complement of \( [a,b] \) whose moments are equal up to degree \( s \).
Then, 
\begin{align*}
    \W(\mu,\nu) 
    \leq  2 (b-a) (1 + \pi^2/2) s^{-1}
    < 12 (b-a) s^{-1}.
\end{align*}
\end{restatable}

We also provide a posteriori error guarantees which may be of practical use.
\mnote{this doesn't have to be \( \Psi_i \)}
\begin{restatable}{theorem}{apostdistances}
\label{thm:apost_distances}
    Let \( \{ [d_i]_j \}_{j=1}^{k} \) and \( \{ [\theta_i]_j \}_{j=1}^{k} \) be the squares of the first component of eigenvectors and the eigenvalues respectively of \( [\vec{T}_i]_{:k,:k} \) from \cref{alg:SLQ}. Then
\begin{align*}
    \KS( \samp{\Psi_i}, \samp{\gq{k}{\Psi_i}})
    &\leq \left\langle \max_{j=1,\ldots, k} [d_i]_j \right\rangle
%    \leq \max_{i\in[\![n_{\textup{v}}]\!], j\in[\![k]\!]} [d_i]_j
    %\max \left\{ 
    %\KS\left(\samp{\Psi_i},  \samp{ \gqlower{k}(\Psi_i) } \right),
    %\KS\left(\samp{\Psi_i},  \samp{ \gqupper{k}(\Psi_i) } \right) \right\}
    ,\\
    \W( \samp{\Psi_i}, \samp{\gq{k}{\Psi_i}})
    \\&\hspace{-5em} \leq\left\langle { \sum_{j=0}^{n} \max\{ [d_i]_j, [d_i]_{j+1} \} ([\theta_i]_{j+1} - [\theta_i]_j) } \right\rangle
    %&\leq \bigg\langle [d_i]_1([\theta_i]_1-a) + [d_i]_k(b-[\theta_i]_n) 
    %\\&\hspace{2em} + \sum_{j=1}^{n-1} \max\{[d_i]_j,[d_i]_{j+1}\} ( [\theta_i]_{j+1} - [\theta_i]_j ) \bigg\rangle.
%    &\leq \W \left( \samp{ \gqlower{k}(\Psi_i) }, \samp{ \gqupper{k}(\Psi_i) } \right)
\end{align*}
where, and for notational convenience, we have defined \( [\theta_i]_0 = a\), \( [\theta_i]_{n+1} = b \), and \( [d_i]_0 = [d_i]_{n+1} = 0 \) for some choice of \( a,b \) such that \( a \leq \lambda_{\textup{min}}[\vec{A}] \) and \( b \geq \lambda_{\textup{max}}[\vec{A}] \).
\end{restatable}
Note that the Wasserstein distance bounds requires knowledge of points \( a,b \) such that \( a \leq \lambda_{\text{min}}[\vec{A}] \) and \( b \geq \lambda_{\text{max}}[\vec{A}] \).
Such bounds can be computed rigorously, both a priori \cite{kuczyski_wozniakowski_92} or a posteriori \cite{parlett_simon_stringer_82}.
In practice, \( \lambda_{\text{min}}([\vec{T}_i]_{:k,:k}) \to \lambda_{\text{min}}[\vec{A}] \) and \( \lambda_{\text{max}}([\vec{T}_i]_{:k,:k}) \to \lambda_{\text{max}}[\vec{A}] \) rapidly, so the \( j=0 \) and \( j=k \) terms can be omitted with negligible effect.

As noted in \cref{thm:exact_lanczos}, the exact CESM can be computed with \( n \) matrix vector products.
However, we also have the following lower bound for a specific class of matrices.
%for any fixed number of matrix vector products, there exist matrices for which the SLQ error will be greater than a constant times the inverse of the number of matrix vector products.
\begin{restatable}{theorem}{wassersteinlb}
    \label{thm:wasserstein_lb}
    For any \( t \in (0,1) \), there exists a matrix \( \vec{A} \) of size \( \lceil (4t)^{-1} \rceil  \) such that if \cref{alg:SLQ} uses fewer than \(  (8t)^{-1} \) matrix vector products, then \cref{alg:SLQ} will output an estimate \( \samp{ \gq{k}{ \Psi_i } } \) satisfying, 
\begin{align*}
    \W( \Phi[\vec{A}], \samp{\gq{k}{\Psi_i} } ) > t I[\vec{A}].
\end{align*}
\end{restatable}

While \cref{thm:apost_distances,thm:wasserstein_lb} involve random variables, the results hold surely and therefore with probability one.

%\begin{proposition}
%    For any even \( k \), there exists exists distribution functions \( \mu \) and \( \nu \), each consisting of \( k/2 \) point masses, supported on \( [-1,1] \), such that \( \W(\mu,\nu) > (2k)^{-1} \).
%\end{proposition}

\section{Analysis and proofs}

For notational convenience, we denote \( \Phi[\vec{A}] \) by \( \Phi \) in proofs.

\subsection{Weighted CESM}

We start with analyzing the weighted CESM.
Note that this analysis is applicable to many algorithms for spectrum approximation, including the KPM.
\begin{lemma}
    \label{thm:psi_beta}
Suppose \( \vec{v} \sim \mathcal{U}(\mathcal{S}^{n-1}) \) and define \( m = n \Phi[\vec{A}](x) \).
Then,
\begin{align*}
    \Psi[\vec{A},\vec{v}](x) \sim \operatorname{Beta} \left( \frac{m}{2}, \frac{n-m}{2} \right).
\end{align*}
\end{lemma}

\begin{proof}
    Let \( \vec{U} = [\vec{u}_1, \ldots, \vec{u}_n] \), where \( \vec{u}_i \) is the \( i \)-th normalized eigenvector of \( \vec{A} \).
    Since \( \vec{U} \) is orthogonal, by the invariance of \( \mathcal{U}(\mathcal{S}^{n-1}) \) under orthogonal transforms, we have that \( \vec{U}^\T \vec{v} \sim \mathcal{U}(\mathcal{S}^{n-1}) \).
    
    We may therefore assume \( \vec{U}^\T\vec{v} \stackrel{\textup{d}}{=} \vec{x} / \| \vec{x} \| \), where \( \vec{x} \sim \mathcal{N}(\vec{0},\vec{I}) \). % and ``\( \stackrel{\textup{d}}{=} \)'' denotes equality in distribution.
    Recall that the \( i \)-th weight of \( \Psi[\vec{A},\vec{v}] \) is given by \( w_i = (\vec{v}^\T \vec{u}_i)^2 \).
    Thus, the \( w_i \) have joint distribution given by, 
\begin{align*}
    w_i \stackrel{\textup{d}}{=} \left(\frac{[\vec{x}]_i}{\|\vec{x}\|} \right)^2 
    = \frac{([\vec{x}]_i)^2}{([\vec{x}]_1)^2+\cdots+([\vec{x}]_n)^2},
\end{align*}
for \( i=1,\ldots, n \).

Recall \( m = n \Phi[\vec{A}](x) \).
Then, 
\begin{align*}
    \Psi[\vec{A},\vec{v}](x)
    = \sum_{j=1}^{m} w_j
    \stackrel{\textup{d}}{=} 
    \frac{([\vec{x}]_1)^2 + \cdots + ([\vec{x}]_m)^2}{([\vec{x}]_1)^2+\cdots+([\vec{x}]_n)^2}.
\end{align*}    
    It is well known that for independent chi-square random variables \( Y \sim \chi_{\alpha}^2 \) and \( Z \sim \chi_{\beta}^2 \) (see, for example, \citep[Section 25.2]{johnson_kotz_balakrishnan_94}),
\begin{align*}
    \frac{Y}{Y+Z}\sim \operatorname{Beta} \left( \frac{\alpha}{2}, \frac{\beta}{2} \right).
\end{align*}
    Thus, since \( ([\vec{x}]_1)^2 + \cdots + ([\vec{x}]_m)^2 \) and \( ([\vec{x}]_{m+1})^2 + \cdots + ([\vec{x}]_n)^2 \) are independent chi-square random variables with \( m \) and \( n-m \) degrees of freedom respectively, \( \Psi[\vec{A},\vec{v}](x) \) is a beta random variable with parameters \( m/2 \) and \( (n-m)/2 \).
\end{proof}

\begin{figure}[tb]
    \centering
    \includegraphics[width=\columnwidth]{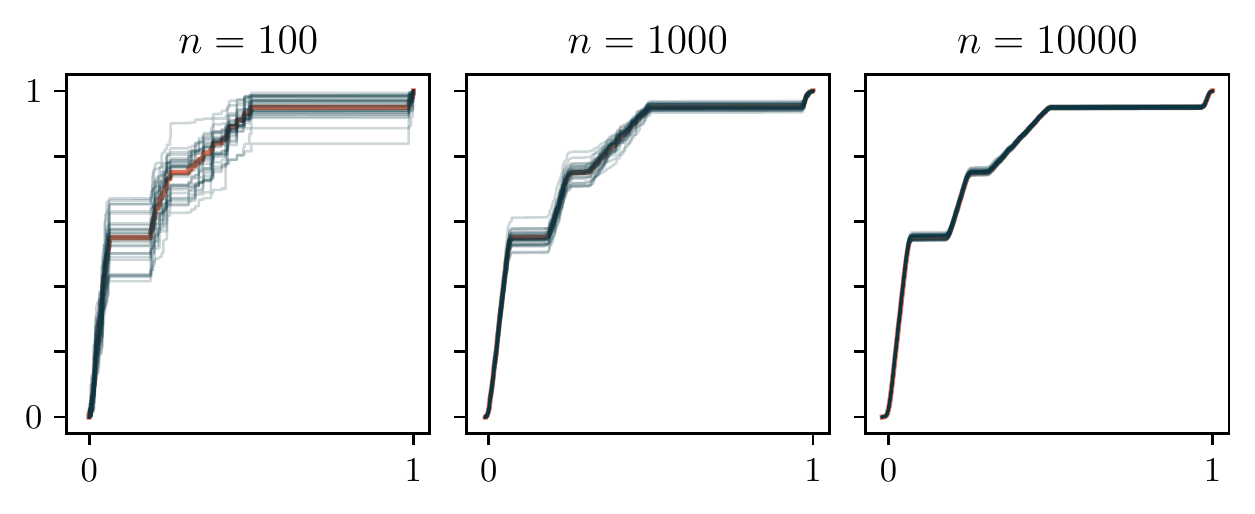}
    \vspace{-2em}
    \caption{
        Concentration of 30 independent samples of the weighted CESM \( \Psi[\vec{A},\vec{v}] \)
        ({\protect\raisebox{.0mm}{\protect\includegraphics[scale=.7]{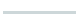}}})
        about the CESM \( \Phi[\vec{A}]\) ({\protect\raisebox{0mm}{\protect\includegraphics[scale=.7]{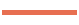}}}) for matrices of different sizes constructed with qualitatively similar spectrums.
        \emph{Remarks}:
        (i) the light lines are samples of a random variable with expectation given by the dark line, 
        (ii) samples of this random variable define cumulative probability densities, and
        (iii) this random variable concentrates exponentially about the CESM as \( n \) increases.
    }
    \label{fig:wCESM}
\end{figure}

As seen in \cref{fig:wCESM}, \( \Psi[\vec{A},\vec{v}](x) \) concentrates about its mean \( \Phi[\vec{A}](x) \) as \( n \) increases.
To understand this more precisely, we introduce the following definition and its consequences.

\begin{definition}
A random variable \( X \) is \( \sigma^2\)-sub-Gaussian if %\( \EE[|X|] < \infty \) and  
\begin{align*}
    \EE\big[ \exp(\lambda (X - \EE[X]))\big]
    \leq \exp \left( \frac{\lambda^2 \sigma^2}{2} \right)
    ,~\forall \lambda \in \mathbb{R}.
\end{align*}
%The smallest value of \( \sigma^2 \) such that \( X \) is \( \sigma^2 \) sub-Gaussian is called the optimal proxy variance of \( X \).
\end{definition}

\begin{restatable}{lemma}{subgaussiansum}
\label{thm:subgaussian_sum}
Suppose \( X \) is \( \sigma^2 \)-sub-Gaussian.
Let \( X_1, \ldots, X_{n_{\textup{v}}} \) be iid samples of \( X \).
Then for all \( t \geq 0 \),
\begin{align*}
    \PP\big[ |\samp{X_i} - \EE[X] | > t \big]
    \leq 2 \exp \left( - \frac{n_{\textup{v}}}{2 \sigma^2} t^2 \right).
\end{align*}
\end{restatable}

\begin{ctheorem}[{\citep[Theorem 1]{marchal_arbel_17}}]
\label{thm:beta_tails}
Suppose \( X \sim \operatorname{Beta}(\alpha,\beta) \).
Then, \( \EE[X] = \alpha / (\alpha + \beta) \), and \( X \) is \( (4(\alpha +\beta + 1))^{-1} \)-sub-Gaussian.
If \( \alpha = \beta \), then there is no smaller \( \sigma^2 \) such that \( X \) is \( \sigma^2 \)-sub-Gaussian.
\end{ctheorem}

With these results in place, the proof of \cref{thm:CESM_estimator_prob} is straightforward.
\begin{proof}[Proof of \cref{thm:CESM_estimator_prob}]
First note that the maximums exist because \( \Phi\) and \( \samp{\Psi_i} \) are right continuous and piecewise constant except at \( \{ \lambda_i[\vec{A}] \}_{i=1}^{n} \).

For any \( x \), let \( m  = m(x) = n \Phi(x) \).
Using \cref{thm:psi_beta,thm:beta_tails,thm:subgaussian_sum} we have that for any \( x \),
\begin{align*}
    &\PP\big[ \left| \Phi(x) - \samp{ \Psi_i(x) } \right| > t \:\big]
    \\&\hspace{3em}\leq 2 \exp \left( - \frac{n_{\textup{v}}}{2 (4(\frac{m}{2}+\frac{n-m}{2}+1) )^{-1}} t^2 \right).
\end{align*}
%Both \( \Psi \) and \( \samp{\Psi_i} \) are right continuous, so the first result follows from a limit argument. 

We also have
\begin{align*}
    &\sup_{x\in\R} |\Phi(x) - \samp{\Psi_i(x)}|
    \\& \hspace{4em}= \max_{i=1,\ldots,n-1} | \Phi(\lambda_i[\vec{A}]) - \samp{\Psi_i(\lambda_i[\vec{A}])} |.
\end{align*}
%\note{TODO: line number}
The second result follows by applying a union bound to the events that the maximum is attained at \( \lambda_i[\vec{A}] \) for each \( i=1,\ldots, n \).
\end{proof}

\subsection{Gaussian quadrature}

We now shift our attention to the approximation of the weighted CESM by a Gaussian quadrature rule.
\begin{definition}
Let \( \mu \) be a distribution function with finite moments up to degree \( 2k-1 \).
The \( k \)-point Gaussian quadrature rule for \( \mu \), is the distribution
\begin{align*}
    \nu(x) = \sum_{j=1}^{k} d_j \bOne[\theta_j \leq x]
\end{align*}
corresponding to nodes \( \{ \theta_j \}_{j=1}^{k} \) and weights \( \{ d_j \}_{j=1}^{k} \) such that the moments of \( \mu \) and \( \nu \) are equal up to degree \( 2k-1 \). 
We denote such a distribution function by \( \gq{k}{\mu} \).
\end{definition}

This definition implies the total mass of a Gaussian quadrature rule must agree with the original distribution, which in the context of computing approximations to the weighted CESM, means that the SLQ approximation remains a probability distribution function.
This property is not retained by other approaches to approximating the weighted CESM such as the KPM.
More generally, \cref{thm:wasserstein_moments} asserts that the Wasserstein distance decays inversely with the number of matching moments.

\mnote{\cite{kong_valiant_17} Proposition 1}
Since, \cref{thm:wasserstein_moments} holds uniformly for all probability distribution functions constant on the complement of \( [a,b] \), we also recall an a posteriori characterization of the closeness of distribution functions with matching moments due to \cite{karlin_shapley_72} but known implicitly far earlier \cite{stieltjes_18}.
Before stating this theorem, we introduce a definition and a resulting lemma.

\begin{definition}
A function \( \gamma \) has a sign change at \( x \) if there exists \( x' < x \) such that \( \gamma(x')\neq 0 \) and \( x = \inf\{t > x' : \gamma(t) \gamma(x') < 0 \} \).
\end{definition}

\begin{restatable}{lemma}{weaksign}
\label{thm:weak_sign}
Suppose \( \gamma \) is weakly increasing on an interval \( (a,b) \).
Then \( \gamma \) and has a sign change at \( x \) if and only if there exists \( x'<x \) such that \( \gamma(x') < 0 \), \( \gamma(y) \leq 0 \) for all \( y\in(a,x) \) and \( \gamma(y) > 0 \) for all \( y\in(x,b) \).
%The \( \gamma \) has at most one sign change on \( (a,b) \).
\end{restatable}

\begin{restatable}[{\citep[Theorem 22.1]{karlin_shapley_72}}]{ctheorem}{momentsCDFinterlace}
\label{thm:moments_CDF}
Suppose \( \mu \) and \( \nu \) are two probability distribution functions constant on the complement of \( [a,b] \) whose moments are equal up to degree \( s \).
Define \( \gamma : [a,b] \to [0,1] \) by \( \gamma(x) = \mu(x) - \nu(x) \).
Then \( \gamma \) is identically zero or changes sign at least \( s \) times.
\end{restatable}

Note that for a probability distribution function, \( \gq{k}{\mu} \) is piecewise constant with \( k \) points of discontinuity.
Using the fact that \( \gq{k}{\mu} \) and \( \mu \) share moments up to degree \( 2k-1 \) along with \cref{thm:moments_CDF}, we immediatley obtain the following bounds on \( \gq{k}{\mu} \) (proved in \cref{sec:proofs} for completeness).

\begin{restatable}{corollary}{gqinterlace}
\label{thm:gq_interlace}
Suppose \( \mu \) is a probability distribution function constant on the complement of \( [a,b] \). % with finite moments up to degree \( 2k-1 \).
Let \( \{ \theta_j \}_{j=1}^{k} \) and \( \{ d_j \}_{j=1}^{k} \) respectively be the nodes and weights of the Guassian quadrature rule \( \gq{k}{\mu} \).
Define \( \gqlower{k}{\mu} \) and \( \gqupper{k}{\mu} \) by
\begin{align*}
    \gqlower{k}{\mu}(x)
    &:= \sum_{j=1}^{k-1} d_{j} \bOne\left[ \theta_{j+1} \leq x \right]
    ,
    \\
    \gqupper{k}{\mu}(x)
    &:= d_1 + \sum_{j=2}^{k} d_{j} \bOne\left[\theta_{j-1} \leq x \right].
\end{align*}
    Then, for all \( x \in [a,b] \),
\begin{align*}
    \gqlower{k}{\mu}(x)
    \leq \mu(x) 
    \leq \gqupper{k}{\mu}(x).
\end{align*}
\end{restatable}

In turn, \cref{thm:gq_interlace} implies bounds on the Wasserstein and Kolmogorov--Smirnov distances between \( \mu \) and \( \gq{k}{\mu} \).

\begin{restatable}{corollary}{gqapost}
\label{thm:gq_apost}
Suppose \( \mu \) is a probability distribution function constant on the complement of \( [a,b] \). % with finite moments up to degree \( 2k-1 \).
Let \( \{ \theta_j \}_{j=1}^{k} \) and \( \{ d_j \}_{j=1}^{k} \) respectively be the nodes and weights of the Guassian quadrature rule \( \gq{k}{\mu} \).
Then
\begin{align*}
    \KS(\mu,\gq{k}{\mu}) 
    &\leq \max_{j=1,\ldots,k} d_j 
    \\
    \W(\mu,\gq{k}{\mu})
    &\leq \sum_{j=0}^{k} \max\{d_j,d_{j+1}\} ( \theta_{j+1} - \theta_j )
\end{align*}
where we define \( \theta_0 = a \), \( \theta_{k+1} = b \), and \( d_0 = d_{k+1} = 0 \).
\end{restatable}

\iffalse
\begin{proof}
Note that
\begin{align*}
    | \mu(x) - \gq{k}{x} | 
    &\leq \max\{ | \gq{k}{\mu}(x) - \gqupper{k}{\mu}(x) |, |\gq{k}{\mu}(x) - \gqlower{k}{\mu}(x) | \} 
    \\&= \max\{ d_j,d_{j+1} \} \bOne[ x\in [\theta_j,\theta_{j+1}) ].
\end{align*}
Thus,
\begin{align*}
    \KS(\mu,\gq{k}{\mu})
    &= \sup_{x} | \mu(x) - \gq{k}{x} |
    \leq \max_{j=1,\ldots, k} d_j. 
\end{align*}
and
\begin{align*}
    \W(\mu,\gq{k}{\mu})
    &= \int_{a}^{b} |\mu(s) - \gq{k}{\mu}(s) | \d{s}
    %&\leq \int_{a}^{b} \max\{ | \gq{k}{\mu}(s) - \gqupper{k}(\mu)(s) |, |\gq{k}{\mu}(s) - \gqlower{k}(\mu)(s) | \} \d{s}
    \\&\leq \sum_{j=0}^{k} \max\{d_j,d_{j+1}\} ( \theta_{j+1} - \theta_j ).
    \tag*{\qedhere}
\end{align*}
\end{proof}
\fi

\begin{figure*}[ht]
    \includegraphics[width=\textwidth]{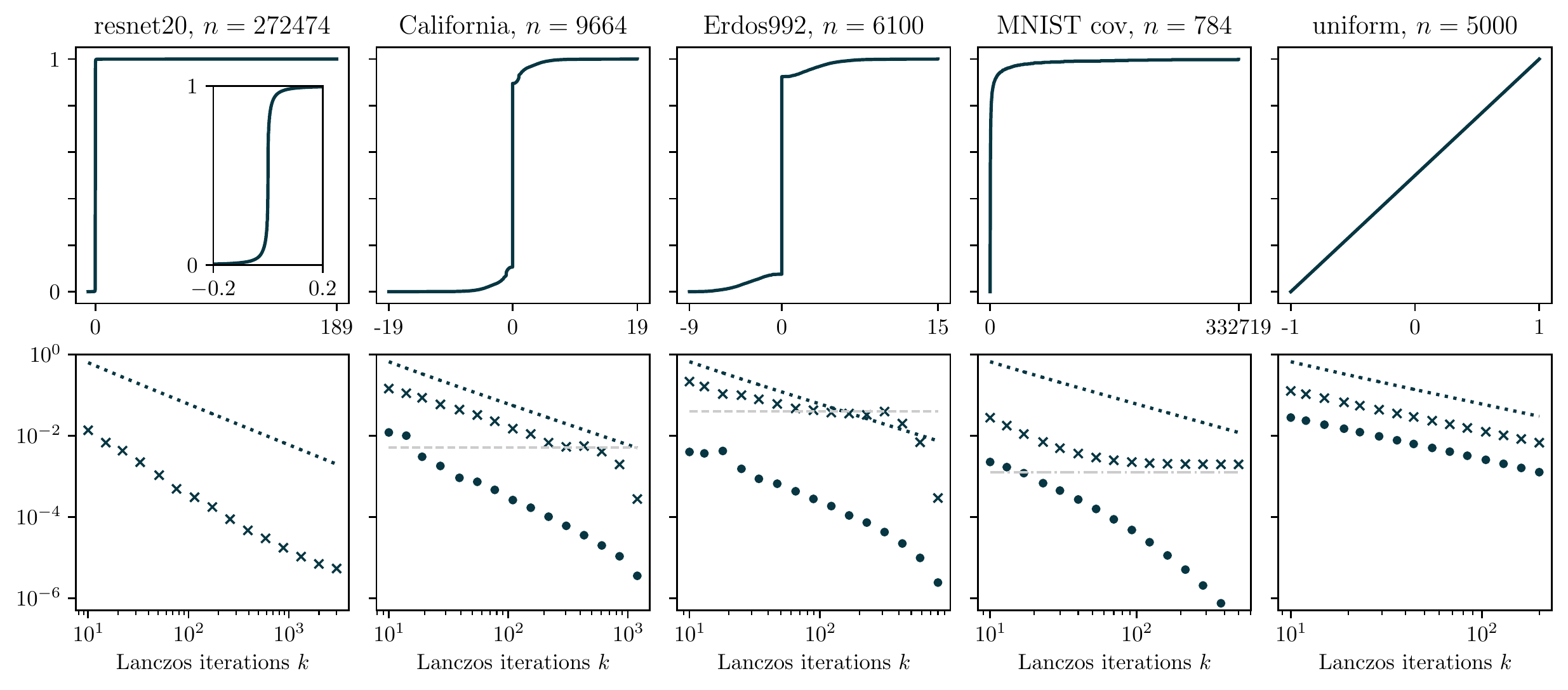}
    \vspace{-2em}
    \caption{
        \emph{Top}: distribution functions
        \emph{Bottom}: Wasserstein error and error bounds. 
        All problems are scaled so that \( I[\vec{A}] = 1 \) for easier comparison.
        From left to right, \( n_{\textup{v}} = 2,6,9,68,11 \) chosen so that \( n_{\text{v}} \) is roughly of size \( O(n^{-1}) \).
        \emph{Legend}: 
        \( \W(\Phi[\vec{A}],\samp{\gq{k}{\Psi_i}} \) 
        ({\protect\raisebox{0mm}{\protect\includegraphics[scale=.7]{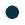}}}),
        bound \( \samp{ \sum_{j=0}^{n} \max\{ [d_i]_j, [d_i]_{j+1} \} ([\theta_i]_{j+1} - \theta_i]_j) } \)
        ({\protect\raisebox{0mm}{\protect\includegraphics[scale=.7]{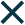}}}),
        bound \( 12 I[\vec{A}] (2k-1)^{-1} \)
        ({\protect\raisebox{0mm}{\protect\includegraphics[scale=.7]{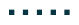}}}),
        \( ( \Phi(d^-) - \Phi(c) ) \: |d-c| \) described in \labelcref{eqn:stagnate}
        ({\protect\raisebox{0mm}{\protect\includegraphics[scale=.7]{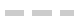}}}),
        \( I[\vec{A}] n^{-1} \)
        ({\protect\raisebox{0mm}{\protect\includegraphics[scale=.7]{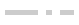}}}).
    }
   \label{fig:wasserstein_bounds}
\end{figure*}

Finally, we note the classical result that Lanczos algorithm computes a Gaussian quadrature rule for \( \Psi[\vec{A},\vec{v}] \) \cite{gautschi_04,golub_meurant_09}.
%\cite{golub_welsch_69,}
\begin{proposition}
\label{thm:lanczos_GQ}
Let \( [\vec{T}]_{:k,:k} \) be the output of \cref{alg:lanczos} run on \( \vec{A} \), \( \vec{v} \) for \( k \) steps.
Then the eigenvalues of \( [\vec{T}]_{:k,:k} \) and the square of the first components of the eigenvectors of \( [\vec{T}]_{:k,:k} \) form a degree \( k \) Gaussian quadrature rule for \( \mu \).
    That is, \( \gq{k}{\Psi[\vec{A},\vec{v}]} = \Psi[[\vec{T}]_{:k,:k},\ehat] \).
\end{proposition}
Since \( k \) is typically much smaller than \( n \), and since \( [\vec{T}]_{:k,:k} \) is tridiagonal, then the exact weighted CESM \( \Psi[[\vec{T}]_{:k,:k},\ehat] \) can be computed directly.
This allows for the efficient computation of Gaussian quadrature rules for \( \Psi[\vec{A},\vec{v}] \).

\subsection{Remaining proofs}

\begin{proof}[Proof of \cref{thm:wasserstein_apriori}]
%We consider the number of samples \( n_{\textup{v}} \) so that \( \W( \samp{ \Psi_i } , \Phi) < t \) and the number of iterations \( k \) so that \( \W( \gq{k}{\Psi_i}, \Psi_i ) < t \) and therefore \( \W( \samp{ \gq{k}{\Psi_i} }, \samp{\Psi_i} ) < t \).
Note that for any probability distribution functions \( \mu \) and \( \nu \) constant on the complement of \( [a,b] \),
\begin{align*}
    \W(\mu,\nu) \leq (b-a) \KS(\mu,\nu).
\end{align*}
For \( i=1,\ldots, n_{\text{v}} \), define \( \Psi_i \) as in \cref{alg:SLQ}.
Then, using \cref{thm:CESM_estimator_prob},
\begin{align*}
    \PP[ \KS(\Phi,\samp{\Psi_i} ) > t/2  ] \leq 2 n \exp (- ( n+2 ) n_{\textup{v}} t^2/4 ),
\end{align*}
so since \( \samp{ \Psi_i } \) and \( \Phi\) are constant on the complement of \( [ \lambda_{\textup{min}}[\vec{A}], \lambda_{\textup{max}}[\vec{A}] ] \),
\begin{align*}
    \PP[ \W(\Phi,\samp{\Psi_i}) > t I[\vec{A}]/2 ] \leq 2 n \exp \left(- n_{\textup{v}}( n+2 ) t^2/4 \right).
\end{align*}

By \cref{thm:wasserstein_moments} and the definition of Gaussian quadrature rule we have, with probability one,
\begin{align*}
    \W(\gq{k}{\Psi_i},\Psi_i) <  12 I[\vec{A}] (2k-1)^{-1}
\end{align*}
for \( i=1,\ldots, n_{\text{v}} \).
Thus, by the triangle inequality, again with probability one,
\begin{align*}
    \W( \samp{\gq{k}{\Psi_i}}, \samp{ \Psi_i } ) <  12 I[\vec{A}] (2k-1)^{-1}.
\end{align*}
Finally, we apply the triangle inequality to obtain,
\begin{align*}
    \W(\Phi,\samp{\gq{k}{\Psi_i}}
    \leq
    \W(\Phi,\samp{\Psi_i})
    + \W(\samp{\Psi_i}, \samp{\gq{k}{\Psi_i}})
\end{align*}
    Setting \( n_{\textup{v}} > 4 (n+2)^{-1} t^{-2} \log(2 n \eta^{-1}) \) and \( k > 12 t^{-1} + \frac{1}{2} \) ensures the sum of the two terms is at most \( I[\vec{A}] t \) with probability at least \( 1-\eta \).
\end{proof}

\begin{proof}[Proof of \cref{thm:apost_distances}]
This is a direct consequence of \cref{thm:gq_apost} and the triangle inequality.
\end{proof}

\begin{proof}[Proof of \cref{thm:wasserstein_lb}]
Let \( \Upsilon(x) = x \) on \( [0,1] \) be the probability distribution function for a uniform density on \( [0,1] \).

First, for any \( K \), non-negative weights \( \{ d_i \}_{i=1}^{K} \) summing to one and ordered points \( \{ \theta_i \}_{i=1}^{K} \) in \( [0,1] \), define \( \{ D_i \}_{i=0}^{K} \) by \( D_i = \sum_{j=1}^{i} d_i \) (where \( D_0 = 0  \)) and consider the functions
\begin{align*}
    \varphi(x) &= \sum_{i=1}^{K} d_i \bOne[ \theta_i \leq x ],
    \\
    \tilde{\varphi}(y) &= \theta_1 + \sum_{i=1}^{K-1} (\theta_{i+1} - \theta_i) \bOne[ D_i \leq y ].
\end{align*}
Note that,
\begin{align*}
    \W(\Upsilon,\varphi) 
    = \int_{0}^{1} | \varphi(x) - x| \d{x}
    = \int_{0}^{1} | \tilde{\varphi}(y) - y| \d{y}.
\end{align*}
Next, define 
\begin{align*}
    \mathcal{C} = &\{ \psi : \psi(0) = 0 , \psi(1) = 1,
    \\
    &\hspace{1em}\psi'(y) = 1 \: \forall y\in(D_i,D_{i+1}) , i=0,\ldots, K \}
\end{align*}
and observe that, because the contribution on each subinterval is independent of other subintervals,
\begin{align*}
    \W(\Upsilon,\varphi) 
    &\geq \min_{\psi\in\mathcal{C}} \int_0^1 | \tilde\varphi(y) - \psi(y)|\d{y}
    \\&= \min_{\psi\in\mathcal{C}} \sum_{i=0}^{K-1} \int_{D_i}^{D_{i+1}} | \tilde\varphi(y) - \psi(y)|\d{y}
    \\&=\sum_{i=0}^{K-1}  \min_{\psi\in\mathcal{C}} \int_{D_i}^{D_{i+1}} | \tilde\varphi(y) - \psi(y)|\d{y}
    \\&=\sum_{i=0}^{K-1} \left( \frac{D_{i+1} - D_i}{2} \right)^2.
\end{align*}
Thus, by the Cauchy--Schwarz inequality,
\begin{align*}
    \frac{1}{4}
    = \left( \sum_{i=0}^{K-1} \frac{d_i}{2} \right)^2
    &\leq \left( \sum_{i=0}^{K-1} 1^2 \right) \left(\sum_{i=0}^{K-1} \left( \frac{D_{i+1} - D_i}{2} \right)^2 \right)
    \\&= K \sum_{i=0}^{K-1} \left( \frac{D_{i+1} - D_i}{2} \right)^2
\end{align*}
so, since \( y \mapsto y \in \mathcal{C} \), we have that
\begin{align*}
    \W(\Upsilon,\varphi) 
    &= \int_{0}^{1} | \tilde{\varphi}(y) - y| \d{y}
    \\&\geq \min_{\psi\in\mathcal{C}} \int_{0}^{1} | \tilde{\varphi}(y) - \psi(y)| \d{y}
    \geq \frac{1}{4K}.
\end{align*}

We now construct a matrix whose CESM has small Wasserstein distance to \( \Upsilon \).
Let \( n = \lceil (4t)^{-1} \rceil \) and define a matrix \( \vec{A} \) with eigenvalues \( \{ (2n)^{-1} + k n^{-1} : k = 0,1,\ldots n-1  \} \). 
By the above argument, noting the the Cauchy--Schwarz inequality is an equality of all terms in the sum are equal, it is clear that,
\begin{align*}
    \W(\Upsilon,\Phi) =
    \frac{1}{4n}
    < t.
\end{align*}
Now, note that \( \samp{\gq{k}{\Psi_i}} \) is of the form of \( \varphi \) with \( K = n_{\text{v}} k \).
Thus, for any \( n_{\text{v}}, k \) such that \( n_{\text{v}} k < (8t)^{-1} \), with probability one,
\begin{align*}
    \W(\Upsilon,\samp{\gq{k}{\Psi_i}})
    \geq \frac{1}{4 n_{\text{v}} k } > 2t.
\end{align*}
Then, using the triangle inequality, again with probability one, 
\begin{align*}
    \W( \Phi,\samp{\gq{k}{\Psi_i}} )
    &\geq \W(\Upsilon,\samp{\gq{k}{\Psi_i}} ) - \W(\Upsilon,\Phi)
    > t.
\end{align*}
Since \( n_{\text{v}} k \) is the number of matrix vector products required by \cref{alg:SLQ}, and \( I(\vec{A}) < 1 \), the result holds.
\end{proof}

Note that this proof constructs two distribution functions with matching moments up to degree \( k \) whose Wasserstein distance is \( \Omega(k^{-1}) \).
This immediately implies that if the output of an algorithm used to approximate distribution functions depends only on the first \( k \) moments, there exist inputs on which the output error has Wasserstein distance \( \Omega(k^{-1}) \).

\iffalse
\begin{proof}[Proof of \cref{thm:wasserstein_lb} (sketch)]
Let \( \Upsilon(x) = x \) be the probability distribution function for a uniform density on \( [0,1] \).
Then, using the triangle inequality,
\begin{align*}
    \W( \Phi,\samp{\gq{k}{\Psi_i}} )
    &\geq \W(\Upsilon,\samp{\gq{k}{\Psi_i}} ) - \W(\Upsilon,\Phi).
\end{align*}

We show that any piecewise constant function with \( K \) points of discontinuity has Wasserstein distance at least \( (4K)^{-1} \) to \( \Upsilon \).
Since \( \{\samp{\gq{k}{\Psi_i}} \) has \( n_{\textup{v}} k \) points of discontinuity, we have \( \W(\Upsilon,\samp{\gq{k}{\Psi_i}} ) > 2t \) with probability one.
    Finally, by constructing \( \vec{A} \) as a matrix with \( \lceil (4t)^{-1} \rceil \) eigenvalues uniformly spaced on \( [0,1] \), we can ensure \( \W(\Upsilon,\Phi) < t\).
\end{proof}
\fi

\section{Numerical verification and discussion}
\label{sec:numerical}
\begin{figure*}[ht]
    \centering
    \includegraphics[width=\textwidth]{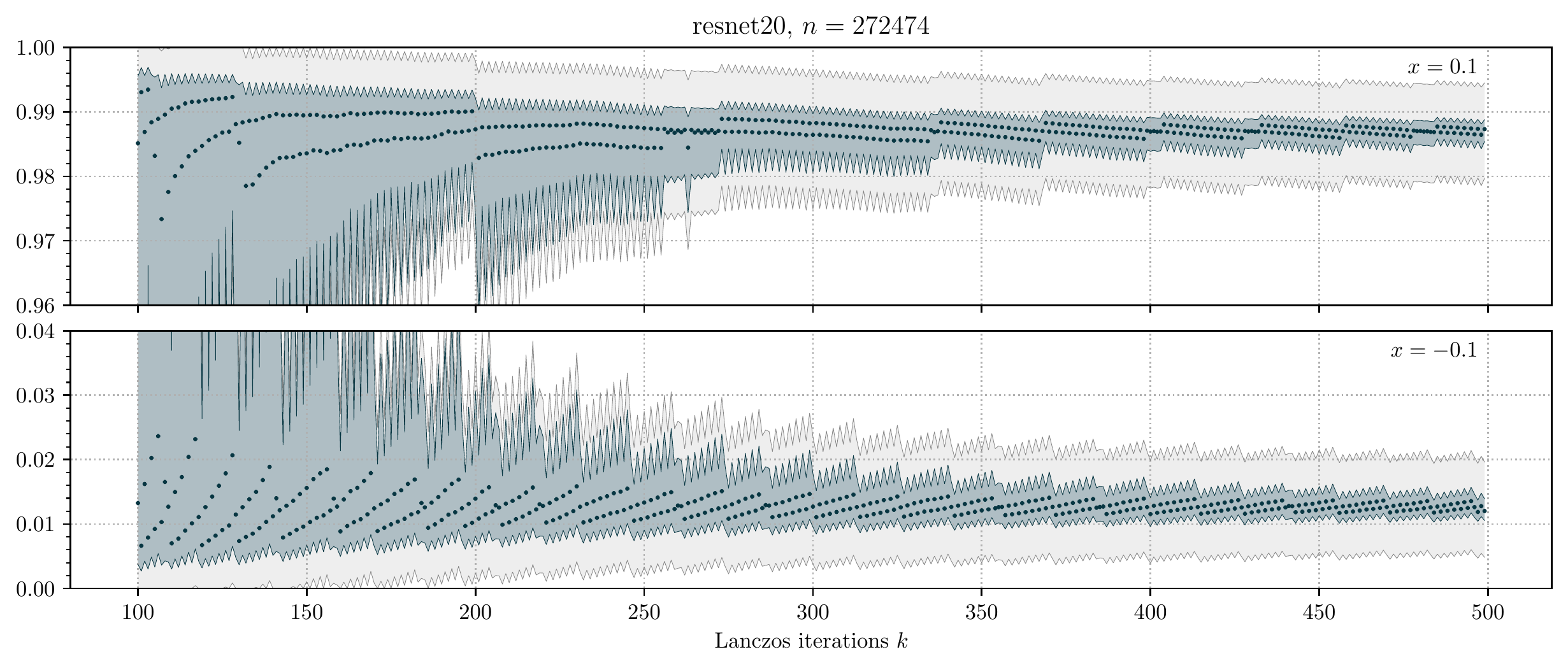}
    \vspace{-2em}
    \caption{
       Bounds for \( \Phi[\vec{A}](x) \) with \( \eta = 0.01 \).
        %({\protect\raisebox{0mm}{\protect\includegraphics[scale=.7]{imgs/legend/solid_red.pdf}}}).
        The average weighted CESM \( \samp{\Psi_i(x)} = \samp{ \Psi(\vec{A},\vec{v}_i)(x) } \) 
        %({\protect\raisebox{0mm}{\protect\includegraphics[scale=.7]{imgs/legend/thin.pdf}}})
        is bounded between the averaged lower and upper bounds \( \samp{ \gqlower{k}{\Psi_i} } \) and \( \samp{ \gqupper{k}{\Psi_i} } \) 
        ({\protect\raisebox{0mm}{\protect\includegraphics[scale=.7]{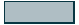}}})
        from \cref{thm:gq_interlace}.
        By \cref{thm:CESM_estimator_prob}, the CESM \( \Phi[\vec{A}](x) \) is within \( t \) of the average weighted CESM with probability at least \( 1-\eta \), and therefore lies betwen \( \samp{ \gqlower{k}{\Psi_i}(x) } - t \) and \( \samp{ \gqupper{k}{\Psi_i}(x) } + t \) 
        ({\protect\raisebox{0mm}{\protect\includegraphics[scale=.7]{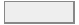}}})
        with this same probability.
        The output of \cref{alg:SLQ} \( \samp{ \gq{k}{\Psi_i} } \) 
        ({\protect\raisebox{0mm}{\protect\includegraphics[scale=.7]{imgs/legend/circle.pdf}}})
        is shown for reference.
        }
    \label{fig:CESM_ave_bound}
\end{figure*}

We demonstrate the effectiveness of our bounds on several test problems.
The convergence of \( \samp{\Psi_i} \) to \( \Phi[\vec{A}]\) is straightforward, so we focus on the convergence of the Gaussian quadrature rules \( \gq{k}{\Psi_i} \) of \( \Psi_i \).

Here, ``resnet20'' is a Hessian for the ResNet20 network \cite{he_zhang_ren_sun_16} trained on the Cifar-10 dataset.
To apply the Lanczos algorithm to this example, we use a slightly modified version of PyHessian \cite{yao_gholami_keutzer_mahoney_20}.
The ``California'' and ``Erdos992'' examples are graph adjacency matrices from the sparse matrix suite \cite{davis_hu_11}, the ``MNIST cov'' example is the covariance matrix of the MNIST training data, and ``uniform'' is a synthetic problem with 5000 eigenvalues evenly spaced between \( -1 \) and \( 1 \).

Our first example studies the global convergence of \( \samp{\gq{k}{\Psi_i}} \) to \( \samp{\Psi_i} \) as the number of Lanczos iterations \( k \) increases.
Specifically, we consider the upper bounds \( \W(\samp{\Psi_i}, \samp{\gq{k}{\Psi_i}}) \leq 12 I[\vec{A}] (2k-1)^{-1} \) from the proof of \cref{thm:wasserstein_apriori} and the bound \( \W(\samp{\Psi_i}, \samp{\gq{k}{\Psi_i}}) \leq \samp{ \sum_{j=0}^{k} \max\{ [d_i]_j, [d_i]_{j+1} \} ([\theta_i]_{j+1} - [\theta_i]_j) } \) from \cref{thm:apost_distances}.
%exponential in \( k \) \cite{kuczyski_wozniakowski_92}.
These bounds, with the true extreme eigenvalues of \( \vec{A} \) used for \( [\theta_i]_0 \) and \( [\theta_i]_{k+1} \) (except for on the ``resnet20'' example where these terms are omitted), are shown illustrated in \cref{fig:wasserstein_bounds} for several test problems. 
Qualitatively, we observe several types of behavior in both the true Wasserstein distance and the bounds.

\iffalse
\begin{theorem}[Theorem 4.2 \cite{kuczyski_wozniakowski_92}]
Suppose \( \vec{A} \) is positive definite and \( \vec{v}\sim \mathcal{U}(\mathcal{S}^{n-1}) \).
Let \( [\vec{T}]_{:k,:k} \) be the output of the Lanczos algorithm run for \( k \) steps on \( \vec{A},\vec{v} \).
Then, for \( \epsilon > \ln(1.648n^{1/2}\eta^{-1})^2 (2k-1)^{-2} \),
\begin{align*}
    \PP\left[ \frac{|\lambda_{\textup{max}}([\vec{T}]_{:k,:k}) - \lambda_{\textup{max}}[\vec{A}]|}{|\lambda_{\textup{max}}[\vec{A}]|} > \epsilon \right] \leq \eta .
\end{align*}
\end{theorem}
\fi

First, when \( k \) is small relative to \( n \), the convergence rate is similar to \( O(k^{-1}) \).
This behavior is especially visible on the ``uniform'' example, where the true CESM is relatively ``smooth'' and aligns with the intuition behind our lower bound \cref{thm:wasserstein_lb}.
On the other hand, as observed on the ``MNIST cov'' example, when \( k \) becomes sufficiently large, the convergence accelerates past \( O(k^{-1}) \).
This is also unsurprising since when \( k = n \), \( \gq{k}{\Psi_i} = \Psi_i \).

Second, we observe the Wasserstein distance bound from \cref{thm:apost_distances} sometimes stagnates.
There are two causes for this. 
The first cause of stagnation, observed on the ``MNIST cov'' example, is due to the fact that the bound from \cref{thm:apost_distances} will never be smaller than \( I[\vec{A}] n^{-1} \).
The second source of stagnation, due to many tightly clustered eigenvalues, is observed in the ``California'' and ``Erdos992'' examples.
To understand this stagnation, suppose there are many eigenvalues clustered near \( x \); i.e. in the interval \( (x-\epsilon,x+\epsilon) \) for some small \( \epsilon \).
Let,
\begin{align*}
    c := \max_{ \lambda \in \Lambda[\vec{A}]} \{ \lambda  < x-\epsilon \} 
    ,&&
    d := \min_{ \lambda \in \Lambda[\vec{A}]} \{ \lambda  > x+\epsilon \}.
\end{align*}
so that there are no eigenvalues in the \( (c,d)\setminus(x-\epsilon,x+\epsilon) \).
Then, if \( |d-c| \gg \epsilon \) and \( k \) is not too large, the Gaussian quadrature rule will have only one node in \( (c,d) \) located very near to \( x \).
As a result, the bound will stagnate near
\begin{align}
    \label{eqn:stagnate}
    ( \Phi(d^-) - \Phi(c) ) \: |d-c|.
\end{align}
If the cluster of eigenvalues are all exactly equal so \( \epsilon \) can be taken to be zero, this stagnation will persist for all \( k \).
If they are not, then eventually the Gaussian quadrature rule will place more nodes in this interval and the bound will recover, as observed in both examples.
In \cref{sec:added_points}, we discuss a heuristic approach to address this stagnation.

Our second example studies bounds for \( \Phi[\vec{A}](x) \) for a fixed value of \( x \).
These can be used to obtain upper and lower bounds for the number of eigenvalues in an interval.
Specifically, we consider the lower and upper bounds \( \samp{\gqlower{k}{\Psi_i}(x)} - t \) and \( \samp{\gqupper{k}{\Psi_i}(x)} + t \), which provide probabilistic upper and lower bounds according to \cref{thm:CESM_estimator_prob,thm:gq_apost}.
In \cref{fig:CESM_ave_bound}, for fixed values of \( x \), we plot the bounds as a function of the number of Lanczos iterations \( k \) for the ``resnet20'' example and as a function of \( x \).
Together, these plots imply that with probability \( 99/100 \), roughly 94-99\% of the eigenvalues are in the interval \( [-0.1,0.1] \).

\section{Outlook}

The analysis in this paper gives rigorous bounds on the accuracy of SLQ for spectrum approximation.
These bounds are suited to the parameter ranges encountered in practice and demonstrate that SLQ is a viable method for spectrum approximation in many applications.
As a result, we hope our analysis will allow practitioners to obtain more precise and theoretically justifiable insights about their applications without the need for heuristics.

More broadly, SLQ and KPM fall into a larger class of spectrum approximation algorithms which approximate iid samples of the weighted CESM using information from Krylov subspaces.
However, the exact relationship between these algorithms has not been fully described, making the tradeoffs between the algorithms murky at best.
In order to shed light on the tradeoffs between these algorithms, a comprehensive treatment providing a \emph{unified} perspective is needed.

%% file: lanczos_spec_appendix.tex
\section{Omitted proofs}
\label{sec:proofs}

In this section, we provide several of the proofs omitted in the main text.

\subsection{Proof of results}
We provide proofs of \cref{thm:gq_interlace,thm:gq_apost}.
These both follow from \cref{thm:moments_CDF}.

\gqinterlace*

\begin{proof}
%    If \( \gq{k}{\mu} = \mu \) the statement is trivially true.
    Suppose \( \gq{k}{\mu} \neq \mu \) and define \( \gamma(x) = \mu(x) - \gq{k}{\mu}(x) \).
    Observe that for any \( j = 1,\ldots, k-1 \), \( \gq{k}{\mu} \) is constant on \( (\theta_j,\theta_{j+1}) \), so \( \gamma \) is weakly increasing on this interval.
    \Cref{thm:weak_sign} states that if \( \gamma \) change signs at some point \( y_j \in (\theta_j,\theta_{j+1}) \) then \( \gamma(x) > 0 \) for all \( x\in (y_j,\theta_{j+1}) \), \( j=1,\ldots, k-1 \) and \( \gamma(x) \leq 0 \) for all \( x\in(\theta_j,y_j) \), so \( \gamma \) cannot change signs at any other point in \( (\theta_j,\theta_{j+1}) \).
     
    Observe further that on \( (a,\theta_1) \), \( \gq{k}{\mu}(x) = 0 \leq \mu(x) \) so \( \gamma(x) \geq 0 \), and on \( (\theta_n,b) \), \( \gq{k}{\mu}(x) = 1 \geq \mu(x) \) so \( \gamma(x) \leq 0 \).
    Thus, by \cref{thm:weak_sign}, no sign changes can occur on these intervals.%\( (a,\theta_1) \) or \( (\theta_n,b) \).

    As a result, the only possible locations for sign changes of \( \gamma \) on \( (a,b) \) are \( \{\theta_j\}_{j=1}^{k} \) and \( \{y_j\}_{j=1}^{k-1} \).
    This is exactly \( 2k-1 \) possible sign changes, so by \cref{thm:moments_CDF}, a sign change must occur at each of these points.
    In particular, 
    \iffalse
    Using this and \cref{thm:weak_sign}, for all \( j=1,\ldots, k-1 \),
    \begin{align*}
        \mu(x) &> \gq{k}{\mu}(x), & x&\in(y_j,\theta_{j+1})
        \\
        \mu(x) &\leq \gq{k}{\mu}(x) , & x&\in(\theta_j,y_j)
    \end{align*}
    and additionally, \fi
    since \( \gamma \) has a sign change at \( \theta_j \),
    \begin{align*}
        \gq{k}{\mu}(\theta_j^-) 
        \leq \mu(\theta_j) \leq \gq{k}{\theta_j}.
    \end{align*}
    Therefore, for \( x\in (\theta_j,\theta_{j+1}) \),
    \begin{align*}
        \gqlower{k}{\mu}(x) =
        \gq{k}{\mu}(\theta_{j}^-) \leq
        \mu(\theta_j) \leq
        \mu(x) 
        &\leq \mu(\theta_{j+1})
        \leq \gq{k}{\mu}(\theta_{j+1})
        = \gqupper{k}{\mu}(x).
        \tag*{\qedhere}
    \end{align*}
\end{proof}

\gqapost*
\begin{proof}
Note that
\begin{align*}
    | \mu(x) - \gq{k}{x} | 
    &\leq \max\{ | \gq{k}{\mu}(x) - \gqupper{k}{\mu}(x) |, |\gq{k}{\mu}(x) - \gqlower{k}{\mu}(x) | \} 
    \\&= \max\{ d_j,d_{j+1} \} \bOne[ x\in [\theta_j,\theta_{j+1}) ].
\end{align*}
Thus,
\begin{align*}
    \KS(\mu,\gq{k}{\mu})
    &= \sup_{x} | \mu(x) - \gq{k}{x} |
    \leq \max_{j=1,\ldots, k} d_j. 
\end{align*}
and
\begin{align*}
    \W(\mu,\gq{k}{\mu})
    &= \int_{a}^{b} |\mu(s) - \gq{k}{\mu}(s) | \d{s}
    %&\leq \int_{a}^{b} \max\{ | \gq{k}{\mu}(s) - \gqupper{k}(\mu)(s) |, |\gq{k}{\mu}(s) - \gqlower{k}(\mu)(s) | \} \d{s}
    \leq \sum_{j=0}^{k} \max\{d_j,d_{j+1}\} ( \theta_{j+1} - \theta_j ).
    \tag*{\qedhere}
\end{align*}
\end{proof}

\subsection{Definitions of random variables}

\begin{definition}
The beta distribution with parameters \( \alpha,\beta \), denoted \( \operatorname{Beta}(\alpha,\beta) \), is defined by density function on \( [0,1] \),
\begin{align*}
    x\mapsto \frac{\Gamma(\alpha + \beta)}{\Gamma(\alpha)\Gamma(\beta)}x^{\alpha-1} (1-x)^{\beta-1}.
    %,&&    \operatorname{B}(\alpha,\beta) = \frac{\Gamma(\alpha)\Gamma(\beta)}{\Gamma(\alpha + \beta)}.
\end{align*}
   
\end{definition}

\begin{definition}
The chi-square distribution with positive integer parameter \( k \), denoted \( \chi_k^2 \), %is distributed as \( \Gamma(k/2,2) \).
is defined by density function on \( [0,\infty) \),
\begin{align*}
    x\mapsto \frac{1}{2^{k/2}\Gamma(k/2)} x^{k/2-1} \exp(-x/2).
\end{align*}
\end{definition}

\subsection{Other proofs}

For reader convenience, we also provide proofs for several of the technical results cited throughout the main portion of the text.

To begin, we provide a standard moment generating function argument to bound the tails of random variables. %\cite{vershynin_18}.
\subgaussiansum*

\begin{proof}

WLOG assume \( \EE[X] = 0 \).
\begin{align*}
    \PP[ n_{\text{v}} \samp{X_i} \geq n_{\text{v}} t] 
    &= \PP[ \exp( \lambda n_{\text{v}}\samp{X_i} ) \geq \exp(\lambda n_{\text{v}} t)]
    \\&\leq \exp(- n_{\text{v}} \lambda t) \EE[ \exp(\lambda n_{\text{v}} \samp{X_i}) ] \tag{Markov}
    \\&= \exp(- n_{\text{v}} \lambda t) \EE[ \exp(\lambda X) ]^{n_{\text{v}}} \tag{iid}
    \\&\leq \exp(- n_{\text{v}} \lambda t) \exp( n_{\text{v}} \lambda^2 \sigma^2 / 2) \tag{sub-Gaussian}
    \\&= \exp(- n_{\text{v}} \lambda t + n_{\text{v}} \lambda^2 \sigma^2 / 2). 
\end{align*}
This expression is minimized when \( \lambda = t / \sigma^2 \) from which we obtain,
\begin{align*}
    \PP[ \samp{X_i} \geq t] \leq \exp \left( - \frac{n_{\text{v}}}{2\sigma^2} t^2 \right). \tag*{\qedhere}
\end{align*}
\end{proof}

\iffalse
Next, we recall an alternate characterization of the Wasserstein distance which is useful in the poof of \cref{thm:wasserstein_moments}.
\begin{lemma}
\label{thm:wasserstein_1lip}
Suppose \( \mu \) and \( \nu \) are two probability distribution functions.
Then,
\begin{align*}
    \W(\mu,\nu) = \sup \left\{ \int f(x) (\d\mu(x) - \d\nu(x)) : | f(y) - f(z) | \leq | y-z| ~\forall y,z\in\R \right\}.
\end{align*}
\end{lemma}
\fi

\wassersteinmoments*

\begin{proof}
Define \( f : [a,b] \to \R \) by
\begin{align*}
    f(x) = - \int_{a}^{x} \operatorname{sign}(\mu(s) - \nu(s)) \d{s}.
\end{align*}
Then, \( f \) is continuous and piecewise linear with slope \( \pm 1 \), and as a result,  \( f \) is of bounded variation and \( |f(y) - f(z)| \leq |y-z| \) for all \( y,z\in[a,b] \).
Moreover, since \( \mu \) and \( \nu \) are weakly increasing bounded functions, they are each of bounded variation, as is the difference \( \mu-\nu \).
    Therefore, we can therefore integrate by parts over the closed interval \( [a,b] \), see for instance \citep[Theorem 3.36]{folland_99}, to obtain
\begin{align*}
    \W(\mu,\nu)
    = \int_{a}^{b} |\mu(s)-\nu(s)| \d{s}
    &= \int_{a}^{b} \operatorname{sign}(\mu(s)-\nu(s)) (\mu(s)-\nu(s)) \d{s}
    \\&= \big[-f(s)(\mu(s)-\nu(s))\big]_{a^-}^b + \int_{a}^{b} f(s) \d(\mu(s)-\nu(s))
    \\&= \int_{a}^{b} f(s) \d(\mu(s)-\nu(s)).
\end{align*}

Next, note that for any polynomial \( p \) of degree \(  < k \),    
\begin{align*}
    \int_{a}^{b} f(s) \d(\mu(s) - \nu(s))
    &= \int_{a}^{b} (f(s) - p(s)) \d(\mu(s)-\nu(s))
    \\&\leq \int_{a}^{b} \|f - p\|_{\infty} | \d\mu(s) - \nu(s) | 
    \\&\leq 2 \| f - p \|_{\infty}. 
\end{align*}

Finally, Jackson's theorem, see for example \citep[Theorem 3.7.2]{han_atkinson_09}, asserts that if \( f \) satisfies \( |f(y) -f(z)| \leq |y-z| \) for all \( y,z\in[a,b] \), then over \( [a,b] \),
\begin{align*}
    \min_{\deg(p) \leq s}  \| f - p \|_{\infty} 
    \leq (b-a) (1+\pi^2/2) s ^{-1}.
\end{align*}
Combining these statements and minimizing over polynomials \( p \) of degree \( \leq s \) yields the result.
\end{proof}

\weaksign*
\begin{proof}
Suppose that \( \gamma \) has a sign change at \( x \) and let \( x' \) be as in the definition.
Then, for any \( t > x' \), since \( \gamma \) is weakly increasing, \( \gamma(x') \leq \gamma(t) \).
This means that \( \gamma(t) \gamma(x') < 0 \) implies \( \gamma(x') < 0 \) and \( \gamma(t) > 0 \).
Since \( x = \inf\{t > x' : \gamma(t) \gamma(x') < 0 \}\), we have \( \gamma(y) > 0 \) for all \( y\in(x,b) \) and \( \gamma(y) \leq 0 \) for all \( y\in(x',x) \) and \( \gamma(x') < 0 \)

The reverse direction follows directly from the definition of sign change.
\end{proof}

Finally, we recall the proof of \cref{thm:moments_CDF}.
\momentsCDFinterlace*
\begin{proof}
Suppose \( \gamma \neq 0 \).
For the sake of contradiction, suppose also that \( \gamma \) has fewer than \( s \) sign changes.
Then, there exists a degree at most \( s-1 \) polynomial \( r \) such that for all \( x \in [a,b] \), \( r(x) \gamma(x) \geq 0 \); i.e. pick \( r \) to have a sign change at every sign change of \( \gamma \).

Thus, since \( \gamma \neq 0 \) is right continuous and \( r \) is continuous, 
\begin{align*}
    \int_{a}^{b} r(x) \gamma(x) \d{x} > 0.
\end{align*}

Let \( R \) be an antiderivative of \( r \).
Then, by integrating by parts over the closed interval \( [a,b] \),
\begin{align*}
    \int_{a}^{b} r(x) \gamma(x) \d{x} 
    = \left[ R(x) \gamma(x) \right]_{a^-}^b - \int_{a}^{b} R(x) \d\gamma(x).
\end{align*}
Since \( \mu \) and \( \nu \) are equal on the compliment of \( [a,b] \),
\begin{align*}
    \left[ R(x) \gamma(x) \right]_a^b 
    = R(b) ( \mu(b) - \nu(b)) - R(a) (\mu(a^-) - \nu(a^-))
    = 0
\end{align*}
and, since \( \mu \) and \( \nu \) share moments up to degree \( s \),
\begin{align*}
    \int_{a}^{b} R(x) \d\gamma(x)
    = \int_{a}^{b} R(x) (\d\mu(x) - \d\nu(x))
    =0 .
\end{align*}

This contradicts the earlier assertion that this integral is non-zero.%, so \( \gamma \) must be zero or have at least \( k-1 \) sign changes.
\end{proof}

\section{Additional numerical experiments}

In \cref{fig:CESM_ave_bound_x} we illustrate the bounds on \( \Phi[\vec{A}](x) \) for varying \( x \) and fixed \( k \).
These are the same quantities as shown in \cref{fig:CESM_ave_bound}, which showed them for fixed \( x \) and varying \( k \).

\begin{figure*}[htb]
    \centering
    \includegraphics[width=\textwidth]{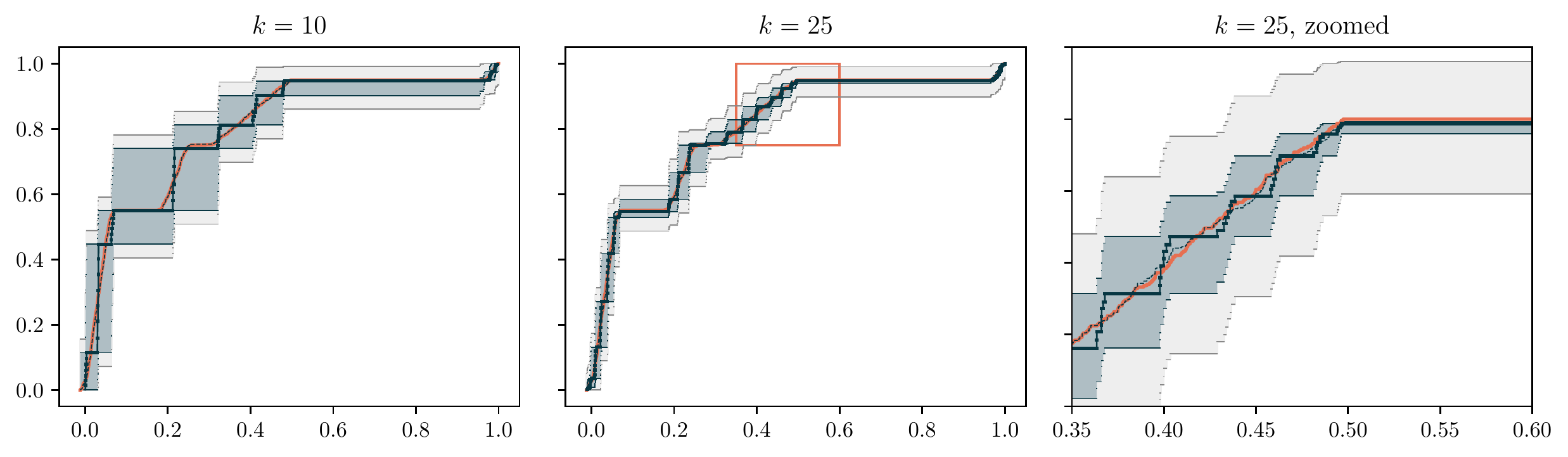}
    \caption{
        Bounds for \( \Phi[\vec{A}](x) \) 
        ({\protect\raisebox{.1mm}{\protect\includegraphics[scale=.75]{imgs/legend/solid_red.pdf}}})
        with \( n_{\text{v}} = 7 \).
        The average weighted CESM \( \samp{\Psi_i(x)} = \samp{ \Psi(\vec{A},\vec{v}_i)(x) } \) 
        %({\protect\raisebox{.1mm}{\protect\includegraphics[scale=.75]{imgs/legend/thin.pdf}}})
        is bounded between the averaged lower and upper bounds \( \samp{ \gqlower{k}{\Psi_i} } \) and \( \samp{ \gqupper{k}{\Psi_i} } \) 
        ({\protect\raisebox{.1mm}{\protect\includegraphics[scale=.75]{imgs/legend/patch.pdf}}})
        from \cref{thm:gq_interlace}.
        By \cref{thm:CESM_estimator_prob}, the CESM \( \Phi[\vec{A}](x) \) is within \( t \) of the average weighted CESM with probability at least \( 1-\eta \), and therefore lies betwen \( \samp{ \gqlower{k}{\Psi_i}(x) } - t \) and \( \samp{ \gqupper{k}{\Psi_i}(x) } + t \) 
        ({\protect\raisebox{.1mm}{\protect\includegraphics[scale=.75]{imgs/legend/patch_grey.pdf}}})
        with this same probability.
        The output of \cref{alg:SLQ} \( \samp{ \gq{k}{\Psi_i} } \) 
        ({\protect\raisebox{.1mm}{\protect\includegraphics[scale=.75]{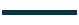}}})
        is shown for reference.
           }
    \label{fig:CESM_ave_bound_x}
\end{figure*}
We make several remarks. 
First, observe that both the average Gaussian quadrature rule \( \samp{\gq{k}{\Psi_i}} \) and the upper and lower bounds \( \samp{\gqupper{k}{\Psi_i}} \) and \( \samp{\gqlower{k}{\Psi_i}} \) are piecewise constant with \( k n_{\text{v}} \) points of discontinuity.
Importantly, however, the points of discontinuity corresponding to different samples concentrate when \( k \) is relatively small compared to \( n \).
This is most visible in the \( k=10 \) plot, which at first glance, appears to only have \( 10 \) points of discontinuity.
This concentration is essentially due to the fact that the weighted CESMs \( \Psi_i \) concentrate about \( \Phi[\vec{A}]\), so moments and therefore the Gaussian quadrature nodes and weights also concentrate.
For further discussions on this behavior we refer readers to \cite{kuijlaars_00,vargas_kulkarni_20}.
Second, observe that the CESM \( \Phi[\vec{A}](x) \) is not necessarily contained between \( \samp{\gqlower{k}{\Phi_i}(x)} \) and \( \samp{\gqupper{k}{\Phi_i}(x)} \).
Rather, the average weighted CESM \( \samp{\Phi_i} \) is containe between these values, and \( \Phi[\vec{A}](x) \) is within \( t \) of \( \samp{\Phi_i} \) with some probability \( \eta \).

In \cref{fig:CESM_ave_bound_ritz_values} we illustrate the same bounds as \cref{fig:CESM_ave_bound} along with a plot of the eigenvalues of \( [\vec{T}_i]_{:k,:k} \). 
We make several remarks. 
First, concentration of the Gaussian quadrature nodes is clearly visible.
Second, the position of these nodes as the iterations change provide insight into the sawtooth behavior of the bounds on \( \samp{\gq{k}{\Psi_i}} \).
Moreover, we observe that the value of \( \samp{\gq{k}{\Psi_i}(x)} \) is very near to the limiting value of \( \samp{\gq{k}{\Psi_i}(x)} \) (which we estimate by computing \( \samp{\gq{3000}{\Psi_i}(x)} \)) whenever there are Guassian quadrature nodes near to \( x \).
Further study of this phenomena may yield useful heuristic CESM estimation.

\afterpage{
\begin{figure}[h]
    \centering
    \includegraphics[width=\textwidth]{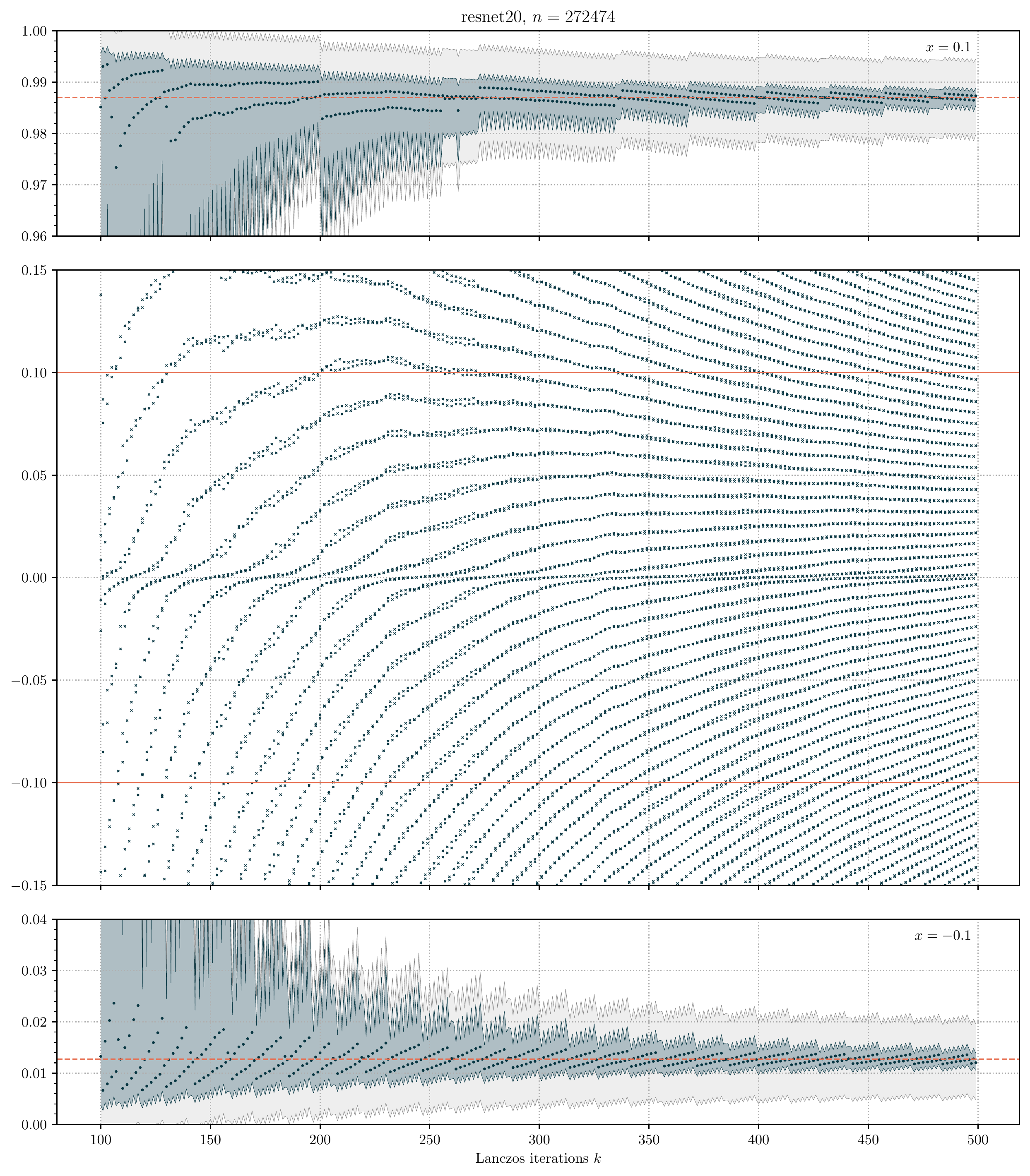}
    \caption{
        \emph{Top/Bottom}: 
        Bounds for \( \Phi[\vec{A}](x) \) when \( x = \pm -.01 \) and \( n_{\text{v}} = 2 \).
        The average weighted CESM \( \samp{\Psi_i(x)} = \samp{ \Psi(\vec{A},\vec{v}_i)(x) } \) 
        %({\protect\raisebox{.1mm}{\protect\includegraphics[scale=.75]{imgs/legend/thin.pdf}}})
        is bounded between the averaged lower and upper bounds \( \samp{ \gqlower{k}{\Psi_i} } \) and \( \samp{ \gqupper{k}{\Psi_i} } \) 
        ({\protect\raisebox{.1mm}{\protect\includegraphics[scale=.75]{imgs/legend/patch.pdf}}})
        from \cref{thm:gq_interlace}.
        By \cref{thm:CESM_estimator_prob}, the CESM \( \Phi[\vec{A}](x) \) is within \( t \) of the average weighted CESM with probability at least \( 1-\eta \), and therefore lies betwen \( \samp{ \gqlower{k}{\Psi_i}(x) } - t \) and \( \samp{ \gqupper{k}{\Psi_i}(x) } + t \) 
        ({\protect\raisebox{.1mm}{\protect\includegraphics[scale=.75]{imgs/legend/patch_grey.pdf}}})
        with this same probability.
        The output of \cref{alg:SLQ} \( \samp{ \gq{k}{\Psi_i} } \) 
        ({\protect\raisebox{.1mm}{\protect\includegraphics[scale=.75]{imgs/legend/circle.pdf}}})
        is shown for reference.
        \emph{Middle}: Eigenvalues of \( [\vec{T}_i]_{:k,:k} \) for \( i=1,2 \).
        \emph{Remarks}: 
        In the middle plot, the eigenvalues of \( [\vec{T}_i]_{:k,:k} \) ({\protect\raisebox{0mm}{\protect\includegraphics[scale=.7]{imgs/legend/x.pdf}}}) are very close for both samples. This aligns with the expected concentration. 
        We also observe that when there is an eigenvalues of \( [\vec{T}_i]_{:k,:k} \) near the value of \( x \) of interest ({\protect\raisebox{0mm}{\protect\includegraphics[scale=.7]{imgs/legend/solid_red.pdf}}}), \( \samp{ \gq{k}{\Psi_i(x}) } \) is near its limiting value \( \samp{\Psi_i(x)} \) (estimated by \( \samp{\gq{3000}{\Psi_i}} \) ({\protect\raisebox{0mm}{\protect\includegraphics[scale=.7]{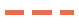}}})).
           }
    \label{fig:CESM_ave_bound_ritz_values}
\end{figure}
\clearpage
}

\subsection{Added points}
\label{sec:added_points}

In \cref{sec:numerical}, we remarked that on spectrums with many clustered eigenvalues, our bounds sometimes encounter plateaus.
In this section, we discuss a heuristic approach to address this issue by introducing eigenvalues with known weights to \( \vec{A} \).

Towards this end, define,
\begin{align*}
    \bar{\vec{A}} = 
    \begin{bmatrix}
        \vec{A} \\
        & y
    \end{bmatrix}
    ,&&
    \bar{\vec{v}}_i
    = \begin{bmatrix}
        \vec{v}_i \\ z
    \end{bmatrix}
\end{align*}
for scalars \( y \) and \( z \).
By increasing the value of \( z \), we can introduce an Gaussian quadrature node near \( y \).
Thus, we may introduce nodes near to locations suspected of having large jumps in the spectrum. 
Often, the origin is such a point, since many matrices may be low rank or close to low rank.

As before, we have that,
\begin{equation}
\label{eqn:added_sample_bd}
\begin{aligned}
    \samp{ \Psi[\bar{\vec{A}},\bar{\vec{v}}_i](x)  }
    &\geq \samp{ \gqlower{k}{\Psi(\bar{\vec{A}},\bar{\vec{v}}_i)}(x) }
    \\
    \samp{ \Psi[\bar{\vec{A}},\bar{\vec{v}}_i](x)  }
    &\leq \samp{ \gqupper{k}{\Psi(\bar{\vec{A}},\bar{\vec{v}}_i)}(x) }
\end{aligned}
\end{equation}
Similarly, \( \Psi(\bar{\vec{A}},\bar{\vec{v}}_i) \) will concentrate strongly around its mean. 
Note, however, that the mean is not \( \Phi(\bar{\vec{A}}) \), since the weight on \( y \) is \( z^2 \) rather than \( 1/(n+1) \).
However, we do have that with probability at least \( 1-2\exp(-n_{\text{v}} (n+2) t^{-2}) \),
\begin{align*}
    \EE[ \Psi[\bar{\vec{A}},\bar{\vec{v}}_i](x) ] 
    &\geq \samp{ \gqlower{k}{\Psi[\bar{\vec{A}},\bar{\vec{v}}_i]}(x) } - t
    \\
    \EE[ \Psi[\bar{\vec{A}},\bar{\vec{v}}_i](x) ] 
    &\leq \samp{ \gqupper{k}{\Psi[\bar{\vec{A}},\bar{\vec{v}}_i]}(x) } + t
\end{align*}
Note that
\begin{align*}
    \Phi[\vec{A}](x) = 
    \EE[ \Psi[\bar{\vec{A}},\bar{\vec{v}}_i](x) ] - z^2 \bOne[y\leq x].
\end{align*}
so
\begin{equation}
\label{eqn:added_full_bd}
\begin{aligned}
    \Phi[\vec{A}](x)
    &\geq \samp{ \gqlower{k}{\Psi[\bar{\vec{A}},\bar{\vec{v}}_i]}(x) } - z^2 \bOne[y\leq x] - t
    \\ 
    \Phi(\vec{A})(x)
    &\leq \samp{ \gqupper{k}{\Psi[\bar{\vec{A}},\bar{\vec{v}}_i]}(x) } -z^2 \bOne[y\leq x] + t.
\end{aligned}
\end{equation}

In \cref{fig:eps_rank_added}, we illustrate the effect of introducing a node at \( 0.05 \), between the large cluster of eigenvalues at the origin, and our point of evaluation, \( x = 0.1 \).
The larger the value of \( z \), the earlier in the iteration the Gaussian quadrature node near \( y \) appears.
However, the larger value of \( z \) also weakens the quality of the lower bounds.

\begin{figure*}[t]
    \includegraphics[width=\textwidth]{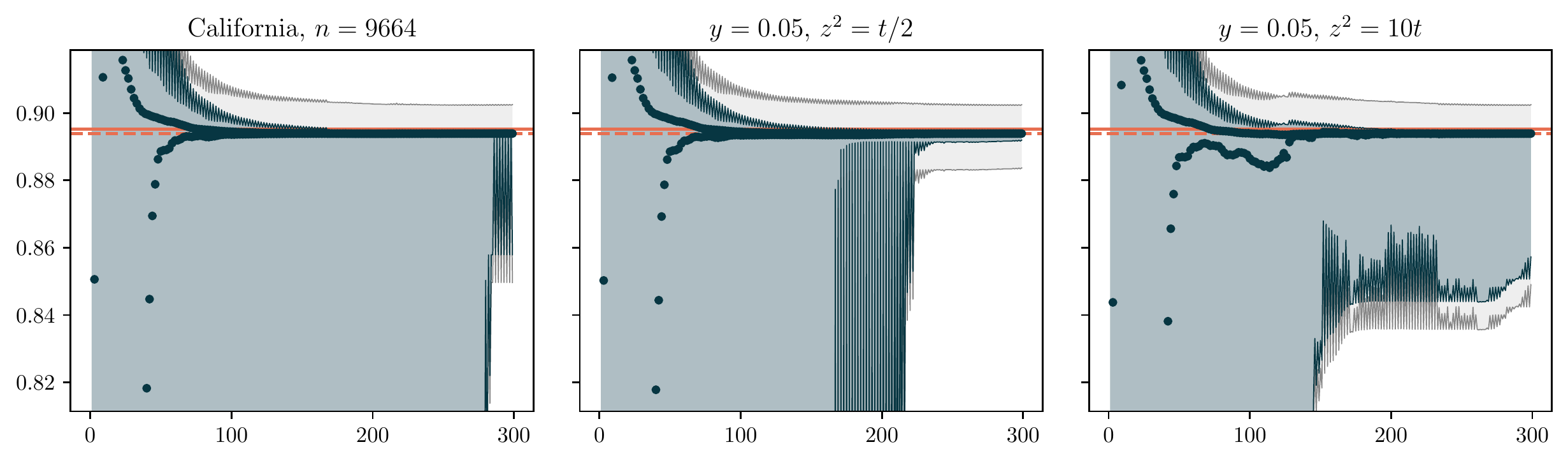}
    \caption{
        Bounds for \( \Phi[\vec{A}](x) \) 
        ({\protect\raisebox{.1mm}{\protect\includegraphics[scale=.75]{imgs/legend/solid_red.pdf}}})
        when \( x = 0.01 \) and \( n_{\text{v}} = 22 \).
        The average weighted CESM \( \samp{\Psi_i(x)} = \samp{ \Psi(\vec{A},\vec{v}_i)(x) } \) 
        %({\protect\raisebox{.1mm}{\protect\includegraphics[scale=.75]{imgs/legend/thin.pdf}}})
        is bounded between the averaged lower and upper bounds \( \samp{ \gqlower{k}{\Psi_i} } \) and \( \samp{ \gqupper{k}{\Psi_i} } \) 
        ({\protect\raisebox{.1mm}{\protect\includegraphics[scale=.75]{imgs/legend/patch.pdf}}})
        from \cref{thm:gq_interlace}.
        By \cref{thm:CESM_estimator_prob}, the CESM \( \Phi[\vec{A}](x) \) is within \( t \) of the average weighted CESM with probability at least \( 1-\eta \), and therefore lies betwen \( \samp{ \gqlower{k}{\Psi_i}(x) } - t \) and \( \samp{ \gqupper{k}{\Psi_i}(x) } + t \) 
        ({\protect\raisebox{.1mm}{\protect\includegraphics[scale=.75]{imgs/legend/patch_grey.pdf}}})
        with this same probability.
        The output of \cref{alg:SLQ} \( \samp{ \gq{k}{\Psi_i} (x)} \) 
        ({\protect\raisebox{.1mm}{\protect\includegraphics[scale=.75]{imgs/legend/circle.pdf}}})
        and \( \samp{ \Psi_i(x) } \)
        ({\protect\raisebox{.1mm}{\protect\includegraphics[scale=.75]{imgs/legend/dashed_red.pdf}}})
        are shown for reference.
    }
    \label{fig:eps_rank_added}
\end{figure*}

\section{Gaussian quadrature convergence for smooth and analytic functions}
\label{sec:GQ_smooth_bounds}

Most existing analysis for SLQ for spectrum approximation has been in terms of smoothed approximations to the empertical spectral measure (ESM) \( \d{\Phi[\vec{A}](x)}/\d{x} \) or the CESM \( \Phi[\vec{A}](x) \) of \( \vec{A} \).
Such smoothed approximations can be obtained by the integral of certain functions against \( \d{\Phi[\vec{A}](x)} \) \cite{lin_saad_yang_16,han_malioutov_avron_shin_17,ghorbani_krishnan_xiao_19}.
Specifically, if \( \upsilon \) approximates a delta function and \( \Upsilon \) approximates a step function, then
\begin{align*}
    \frac{\d{\Phi[\vec{A}](x)}}{\d{x}}
    \approx \int \upsilon(x) \d{\Phi[\vec{A}](x)}
    ,&&
    \Phi[\vec{A}](x) \approx \int \Upsilon(x) \d{\Phi[\vec{A}](x)}.
\end{align*}
Defining \( \Psi_i = \Psi(\vec{A},\vec{v}_i) \) for \( \{ \vec{v}_i \} \stackrel{\text{iid}}{\sim} \mathcal{U}(\mathcal{S}^{n-1}) \), the smoothed ESM/CESM can therefore respectively be approximated by
\begin{align*}
    \left\langle \int \upsilon(x) \d{\gq{k}{\Psi_i}(x)} \right\rangle
    ,&&
    \left\langle \int \upsilon(x) \d{\gq{k}{\Psi_i}(x)} \right\rangle.
\end{align*}
If these \( \upsilon \) and \( \Upsilon \) are smooth or analytic, the following well known bounds apply.

\begin{ctheorem}[{\citep[Section 6.2]{golub_meurant_09}}]
    \label{thm:gq_smooth}
    Let \( \mu \) be a probability distribution function constant on the compliment of \( [a,b] \) and let \( p_k(s) \) be the \( k \)-th monic orthogonal polynomial of \( \mu \).
    Then, if \( f \) is \( 2k \) times differentiable, for some \( \xi \in [a,b] \),
    \begin{align*}
        \int f(s) \d\mu(s) - \int f(s) \d\gq{k}{\mu}(s)
        = \frac{f^{(2k)}(\xi)}{(2k)!} \int \prod_{i=1}^{k} p_k(s)^2 \d\mu(s).
    \end{align*}
\end{ctheorem}

\iffalse
\xnote{
\begin{remark}
In the case that \( f \) is \( 2k \) times differentiable and all derivatives have the same sign on the support of the original measure, then on can obtain upper and lower bounds for Gaussian quadrature using the related Gauss--Lobotto and Gauss--Radau quadrature rules.
These rules can be computed easily from information obtained by the Lanczos algorithm \cite{golub_meurant_09}.
\end{remark}
}
\fi

\begin{ctheorem}[{\citep[Theorem 4.2]{ubaru_chen_saad_17}}]
    \label{thm:gq_analytic}
Let a function \( f \) be analytic in \( [-1,1] \) and analtically continuable in the open Bernstein ellipse \( E_{\rho} \) with foci \( \pm 1 \) and sum of major and minor axis equal to \( \rho > 1 \), where it satisfies \( |g(z)| \leq M_{\rho} \).
Then, the \( k \)-step Lanczos quadrature satisfies,
\begin{align*}
    \left| \int f(s) \d\Psi(s) - \int f(s) \d\gq{k}{\Psi}(s) \right|
    \leq \frac{4 M_{\rho}}{(\rho^2-1) \rho^{2(k-1)}}
\end{align*}
where \( \Psi = \Psi[\vec{A},\vec{v}] \) for some \( \vec{A} \) with eigenvalues in \( [-1,1] \) and a possibly deterministic right hand size \( \vec{v} \).
\end{ctheorem}

At first glance, \cref{thm:gq_smooth,thm:gq_analytic} suggest much faster convergence than \cref{thm:wasserstein_apriori}.
However, in practice, constants hidden by big-\( O \) notation are important.
Typical choices for \( \upsilon \), such as Gaussian densities, grow rapidly away from the real axis, meaning \( M_{\rho} \) grows as \( \rho \) becomes larger.
Similarly, the width of the interval over which \( \Upsilon \) goes from being near 0 to being near 1 is closely related to \( \rho \) and \( M_{\rho} \).
Indeed, in order to closely approximate a discontinuous function, the analytic approximations are necessarily poorly behaved in the complex plane.
This typically manifests as singularities which approach the real axis as the width over which the jump occurs decreases, which in turn, forces \( \rho \to 1 \) and \( M_{\rho} / (\rho^2 - 1) \to \infty \).

As a result, even though SLQ approximations to the smoothed EMS/CESM have exponential convergence, the number of iterations predicted by the bounds is typically not useful for the smoothing parameters used in practice.
As an explicit example, suppose \( \upsilon(s) = \frac{1}{\sqrt{2\pi \sigma^2}} \exp( - (s-x)^2 / 2\sigma^2 ) \) is a Gaussian density function, and that we wish to evaluate the smoothed ESM at \( x= 0 \).
Then, using that \( i(\rho + \rho^{-1})/2 \) is the uppermost point of \( E_{\rho} \),
\begin{align*}
    M_{\rho} \geq 
    \left| \upsilon\left( i \frac{\rho+\rho^{-1}}{2} \right) \right| 
    \geq \left| \upsilon\left( i \frac{\rho}{2} \right) \right|
    = \frac{1}{\sqrt{2\pi \sigma^2}} \exp \left( \frac{\rho^2}{8\sigma^2}  \right).
\end{align*}
Thus, to reduce the bound from \cref{thm:gq_analytic} to a fraction \( c \) of the maximum possible value of the smoothed density, \( \upsilon(0) \), we require
\begin{align*}
    \frac{4M_{\rho}}{(\rho^2-1)\rho^{2(k-1)}} < c \upsilon(0).
\end{align*}
Solving for \( k \) we have,
\begin{align*}
    k > 1 + \frac{ \ln\left( \frac{4 M_{\rho}}{c\upsilon(0)(\rho^2-1)}\right)}{2\ln(\rho/2)}
    = 1 + \frac{ \rho^2 / 8\sigma^2 +  \ln(4) - \ln(c\:(\rho^2-1))}{2\ln(\rho)}
\end{align*}
If \( \sigma^2 = 10^{-5} \) as suggested in \cite{ghorbani_krishnan_xiao_19}, then \( k \) must be exceedingly large for the bound to provide any meaningful convergence gurantee; even for an extremely weak bound of \( c = 1/2 \), we require \( k > 33900 \) iterations.
On the other hand, our bounds, while only inversely dependent on \( k \), are applicable to the parameter ranges found in practice; i.e. for a small number of Lanczos iterations.

\section{The Lanczos algorithm in finite precision}
\label{sec:finite_precision}

\begin{figure*}[t]
    \includegraphics[width=\textwidth]{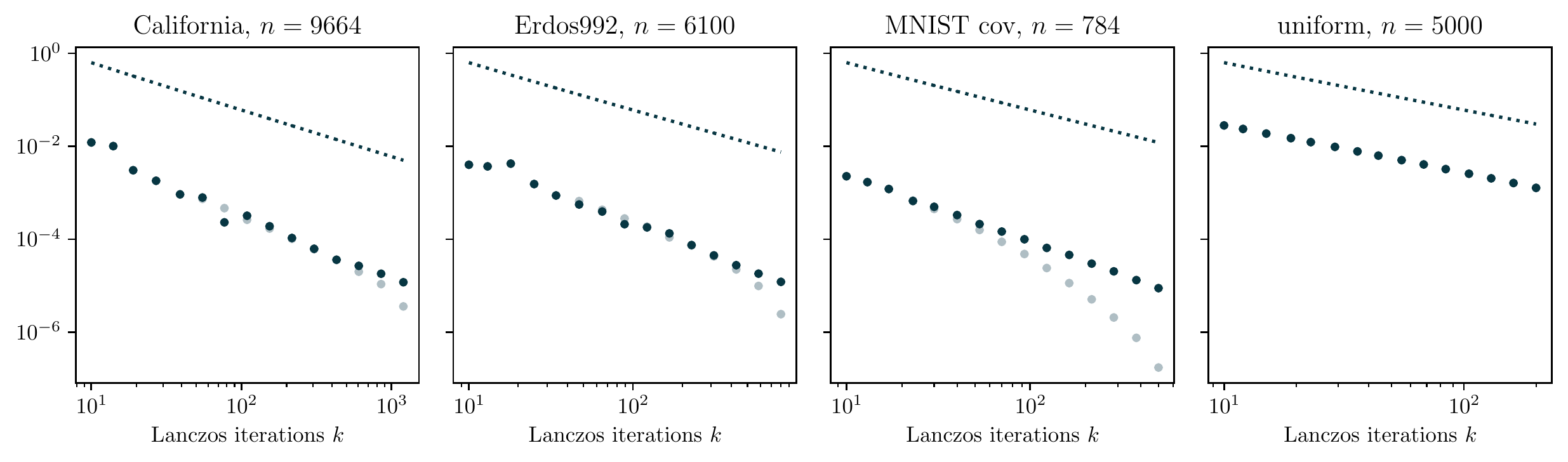}
    \caption{
        Wasserstein error \( \W(\Phi[\vec{A}],\samp{\gq{k}{\Psi_i}} \) (in finite preicsion without reorthogonalization) 
        ({\protect\raisebox{0mm}{\protect\includegraphics[scale=.7]{imgs/legend/circle.pdf}}}).
        All problems are scaled so that \( I(\vec{A}) = 1 \) for easier comparison.
        For comparison, the errors with reorthogonalization 
        ({\protect\raisebox{0mm}{\protect\includegraphics[scale=.7]{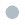}}})
        from \cref{fig:wasserstein_bounds} and the bound \( 12 I[\vec{A}] (2k-1)^{-1} \)
        ({\protect\raisebox{0mm}{\protect\includegraphics[scale=.7]{imgs/legend/dotted.pdf}}}),
        are also included.
    }
    \label{fig:fin_prec}
\end{figure*}

In finite precision, Lanczos may behave very differently than in exact arithmetic \cite{greenbaum_97,meurant_strakos_06}.
The most visible effects are a loss of orthogonality in the basis vectors \( [\vec{Q}]_{:,:k} \) and the appearence of multiple tightly clustered eigenvalues in \( [\vec{T}]_{:k,:k} \) all approximating isolated eigenvalues of \( \vec{A} \) (in exact arithmetic, each eigenvalue of \( \vec{A} \) would be approximated by at most one eigenvalue of \( [\vec{T}]_{:k,:k} \)).
This is due primarily to a loss of orthogonality in the basis vectors \( \{ \vec{q}_i \}_{i=1}^{k} \) caused by orthogonalizing \( \tilde{\vec{q}}_{i+1} \) against only \( \vec{q}_i \) and \( \vec{q}_{i-1} \) which occurs precisely when an eigenvalue of \( [\vec{T}]_{:k,:k} \) has converged to an eigenvalue of \( \vec{A} \) \cite{paige_76,paige_80}.

The most straightforward solution is to reorthogonalize against all previous basis vectors at each step.
This maintains orthogonality to near machine precision, albeit at the cost of increasing the runtime to \( O(k (T_{\text{mv}} + n + k)) \) and the required storage to \( O(nk) \).
Other more computationally efficient schemes, which orthogonalize selevatively against the basis vectors corresponding to converged eigenvalues of \( [\vec{T}]_{:k,:k} \) have also been studied \cite{parlett_scott_79}.

In the context of iterative algorithms, the repeated eigenvalues found by Lanczos often manifests as a delay of convergence; i.e. and increase in the number of iterations required to reach a given accuracy.
This behavior is observed in SLQ and illustrated in \Cref{fig:fin_prec}.
Note that the delay is worse on problems where there are isolated eigenvalues, especically in the upper spectrum.
This is because, in exact arithmetic, Lanczos will quickly find these eigenvalues, orthogonality will be lost, and multiple iterations will be wasted finding repeated copies of the eigenvalue.
On problems such as ``uniform'', where there are no isolated eigenvalues and \( k \ll n \), Lanczos behaves nearly identically in finite precision and exact arithmetic. 

As noted by \cite{papyan_19}, even in finite precision SLQ still seems to converge. 
We now give a high level explanation for this phenomenon.

Under certain conditions, \cite{greenbaum_89} shows that the matrix \( [\vec{T}]_{:k,:k} \) obtained by an implementation of the Lanczos algorithm run in finite precision can be viewed as the output the Lanczos algorithm run in exact arithmetic on a certain ``nearby'' problem. 
Loosely speaking, the necessary conditions for the analysis of \cite{greenbaum_89} to apply are that for all \( i \leq k \),
\begin{enumerate}[label=(\roman*),nolistsep]
    \item the finite precision Lanczos vectors are nearly unit length: \( \| \vec{q}_i \| \approx 1 \)
    \item that they are close locally orthogonal: \( \beta_i \vec{q}_i^\T \vec{q}_{i-1} \approx 1 \)
    \item that they approximately satisfy the three term recurrence: \( \vec{A}\vec{q}_{i} \approx \beta_{i} \vec{q}_{i+1} + \alpha_{i} \vec{q}_{i} +  \beta_{i-1} \vec{q}_{i-1} \).
\end{enumerate}
Such conditions are satisfied by \cref{alg:lanczos} (without reorthogonalization) in finite precision \cite{paige_76,paige_80}.

If these conditions are met, \cite{greenbaum_89} shows that there exists a \( N\times N \) matrix \( \bar{\vec{A}} \) and vector \( \bar{\vec{v}} \) such that Lanczos run on \( \bar{\vec{A}} \), \( \bar{\vec{v}} \) in exact arithmetic for \( k \) steps produces \( [\vec{T}]_{:k,:k} \) and, that
\begin{enumerate}[label=(\roman*),nolistsep]
    \item  
        Eigenvalues of \( \bar{\vec{A}} \) clustered near to those of \( \vec{A} \):
        for any \( j \in \{1,\ldots, N \} \), there exists \( i\in \{1,\ldots, n\}  \) such that 
    \begin{align*}
        \lambda_j(\bar{\vec{A}}) \approx \lambda_i(\vec{A}).
    \end{align*}

    \item 
        The sum of square of first components of eigenvectors of clusters of \( \bar{\vec{A}} \) are near to the square of the projections of \( \vec{v} \) onto the eigenvectors of \( \vec{A} \):
        for an eigenvalue \( \lambda_i(\vec{A}) \)
        \begin{align*}
            w_i \approx \sum_{j\in S} \bar{w}_j
        \end{align*}
        where \( S_i \) is the set of indices such that \( \lambda_j(\bar{\vec{A}}) \approx \lambda_i(\vec{A}) \) for all \( j\in S \).
\end{enumerate}

Together, these conditions imply that
\begin{align}
    \label{eqn:finprec_wCESM}
    \Psi[\vec{A},\vec{v}]
    \approx \Psi(\bar{\vec{A}},\bar{\vec{v}}).
\end{align}
Since the \( [\vec{T}]_{:k,:k} \) produced by the finite precision computation corresponds to an exact Gaussian quadrature rule for \( \Psi(\bar{\vec{A}},\bar{\vec{v}}) \), the theory developed in this paper applies to the exact computation with \( \bar{\vec{A}} \) and \( \bar{\vec{v}} \).
Of course, what the ``\( \approx \)'' in \labelcref{eqn:finprec_wCESM} means depend on the details of both the analysis in \cite{greenbaum_89} and implementation of the Lanczos algorithm used. 
If, for instance, the Wasserstein distance between \( \Psi[\vec{A},\vec{v}] \) and \( \Psi(\bar{\vec{A}},\bar{\vec{v}}) \) is not reasonably small for the precision used, then knowing that the \( [\vec{T}]_{:k,:k} \) corresponds to an exact computation is with \( \bar{\vec{A}},\bar{\vec{v}} \) is not useful in terms of our bounds.